\documentclass[11pt]{article}
\usepackage{amsmath,amsthm,amscd,amssymb}
\usepackage{latexsym}
\usepackage{tikz}
\usepackage{mathrsfs}

\usepackage{yfonts}

\newcommand{\bbN}{{\mathbb{N}}}

\newcommand{\bbR}{{\mathbb{R}}}

\newcommand{\bbZ}{{\mathbb{Z}}}

\newcommand{\bbC}{{\mathbb{C}}}

\newcommand{\gothic}[1]{\textgoth{#1}}

\newcommand{\lozr}{
--++(1,1)--++(1,-1)--++(-1,-1)
--++(-1,1)[fill=white]}

\newcommand{\lozd}{--++(1,-1)--++(0,2)--++(-1,1)--++(0,-2)
[fill=white!30!black]}

\newcommand{\lozu}{--++(1,1)--++(0,2)--++(-1,-1)--++(0,-2)
[fill=white!70!black]}

\newcommand{\no}{\nonumber}

\newcommand{\supp}{\text{\rm{supp}}}

\newcommand{\bi}{\bibitem}

\newcommand{\eps}{\varepsilon}

\newcommand{\beq}{\begin{equation}}

\newcommand{\eeq}{\end{equation}}

\newcommand{\ba}{\begin{align}}

\newcommand{\ea}{\end{align}}

\DeclareMathOperator{\Tr}{Tr}

\DeclareMathOperator*{\EE}{{\mathbb{E}}}

\DeclareMathOperator*{\Var}{Var}

\allowdisplaybreaks

\numberwithin{equation}{section}

\newtheorem{theorem}{Theorem}[section]

\newtheorem{proposition}[theorem]{Proposition}

\newtheorem{lemma}[theorem]{Lemma}

\newtheorem{corollary}[theorem]{Corollary}

\theoremstyle{definition}

\newtheorem{example}[theorem]{Example}

\theoremstyle{remark}

\newtheorem{remark}{Remark}[section]

\title{Central Limit Theorems for Biorthogonal Ensembles and Asymptotics of Recurrence Coefficients}
\author{Jonathan Breuer\footnote{ Einstein Institute of Mathematics, The Hebrew University of Jerusalem,
Jerusalem 91904, Israel. E-mail: jbreuer@math.huji.ac.il. Supported in part by the US-Israel Binational Science Foundation (BSF) Grant no.\ 2010348 and by the Israel Science Foundation (ISF) Grant no.\ 1105/10.} \and Maurice Duits\footnote{Department of Mathematics, Royal Institute of Technology (KTH), Lindstedtsv\"agen 25, SE-10044 Stockholm, Sweden. Supported in part by the grant KAW 2010.0063 from the Knut and Alice Wallenberg Foundation and by the Swedish Research Council (VR) Grant no. 2012-3128. }}

\date{}

\begin{document}
\maketitle 
\begin{abstract}
We study fluctuations of linear statistics corresponding to smooth functions for certain biorthogonal ensembles.  We study those biorthogonal ensembles for which  the underlying biorthogonal family  satisfies a finite term recurrence and describe the asymptotic fluctuations using right limits of the recurrence matrix. As a consequence, we show that whenever the right limit is a Laurent matrix, a Central Limit Theorem holds. We will also discuss the implications for orthogonal polynomial ensembles. In particular, we obtain a Central Limit Theorem  for the orthogonal polynomial ensemble associated with any measure belonging to the Nevai Class of an interval. Our results also extend previous results on Unitary Ensembles in the one-cut case. Finally, we will illustrate our results by deriving Central Limit Theorems for the Hahn ensemble for lozenge tilings of a hexagon and for the Hermitian two matrix model.
\end{abstract}

\section{Introduction}

Eigenvalues of random matrices typically tend to an equilibrium configuration as the size of the matrix tends to infinity. The fluctuations of the empirical measure around this equilibrium have been an important topic of study in random matrix theory  (we refer to \cite{AGZ,F,PS,Handbook} as general references).  In particular, such fluctuations can often be described by certain Central Limit Theorems. This paper is concerned with studying this phenomenon in the general context of orthogonal polynomial ensembles \cite{K} and, even more generally, biorthogonal ensembles \cite{Bor} for which the underlying biorthogonal family satisfies a recurrence relation. The family of biorthogonal ensembles is a rather large family of determinantal point processes which includes many models from probability and combinatorics, such as non-colliding random walks, growth models and last passage percolation, as well as eigenvalues of certain invariant random matrix ensembles (for a review of various important models leading to orthogonal polynomial ensembles see, e.g., \cite{K}; for examples of the generalization to biorthogonal polynomials see \cite{Bor}).  

Let $\mu$ be a Borel measure on $\bbR$. For two families of linearly independent functions $\mathcal \{\psi_j\}_{j \in \bbN}$ and $\mathcal \{\phi_j\}_{j \in \bbN}$ the \emph{biorthogonal ensemble} of size $n$ is defined as the probability measure on $\bbR^n$ proportional to
\begin{equation}\label{eq:defBO}
\det \left(\psi_{j-1}(x_i)\right)_{i,j=1}^n \det \left(\phi_{j-1}(x_i)\right)_{i,j=1})^n {\rm d} \mu(x_1) \cdots {\rm d}\mu(x_n).
\end{equation}
By Gram-Schmidt and the fact that determinants are linear in each of the columns, we can assume without loss of generality that the families are biorthogonal in the sense that
\begin{equation}\label{eq:orthogonality}
\int \psi_j(x) \phi_k(x) {\rm d}\mu(x)= \delta_{jk}.
\end{equation}
The term `biorthogonal ensemble' was introduced by Borodin in \cite{Bor}.  

An important subclass is that of the \emph{orthogonal polynomials ensembles}. In this case we take  $\psi_j=\phi_j=p_j$, a polynomial of degree $j$. In other words, by \eqref{eq:orthogonality} the polynomials are the orthogonal polynomials with respect to $\mu$
$$\int p_k(x) p_l(x) {\rm d} \mu(x)=\delta_{kl}.$$
We will discuss the implications of our main results for orthogonal polynomials ensembles  below. 

Let $x_1,\ldots,x_n$ be $n$ random point chosen from a biorthogonal ensemble \eqref{eq:defBO}. We are interested in the asymptotic behavior of these points in the limit $n\to \infty$. In many situations, the points converge to a deterministic configuration almost surely. That is, we have weak convergence of the empirical measure
$$\frac{1}{n} \sum_{j=1}^n \delta _{x_j} \to \nu,$$
almost surely, for some special measure $\nu$ on $\bbR$. A famous example is the Gaussian Unitary Ensemble where $\nu$ is the semi-circle law. More generally, in the case of orthogonal polynomial ensembles, $\nu$ is often given by  the minimizer of a  particular energy functional. See \cite{AGZ} for a general discussion.

In this paper we will be interested in the fluctuations of the empirical measure around its mean.  The fluctuations we shall study are those of linear statistics associated with sufficiently smooth functions. Given a function $f$ on $\bbR$, and a biorthogonal ensemble of size $n$, the linear statistic associated with $f$, $X_f^{(n)}$, is defined by 
$$ 
X_f^{(n)} = \sum_{j=1} ^n f(x_j). 
$$  
Linear statistics form a powerful tool in the study of determinantal point processes (and for  general $\beta$ analogues)  and a substantial amount of effort has been devoted to studying Central Limit Theorems that they satisfy. The amount of literature on the subject is too great for us to be able to review it here. A partial list of references that are relatively close to our setting is \cite{Borad,Borad2,Borad3,DP,HKPV,Jannals,Jduke,KS,SoshInf,SoshCLT}. If $f$ is sufficiently smooth, one often has  
\beq \label{eq:RMTCLT}
X_f^{(n)} - \EE X_f^{(n)} \to N(0,\sigma_f^2), \quad \text{as } n\to \infty,
\eeq
in distribution, for some particular value of $\sigma_f^2$.

\begin{remark}
Note the difference between this formulation and the standard Central Limit Theorem for sums of independent random variables, in that \eqref{eq:RMTCLT} lacks the normalization by $\sqrt{n}$. This is due to the effective repulsion between the random eigenvalues, which is typical for models of random matrix theory. For an excellent discussion of this phenomenon we refer to \cite{Diac}. Part of \eqref{eq:RMTCLT} is that the variance of $X_f^{(n)}$ remains bounded as $n\to \infty$. This can only be expected to hold for sufficiently smooth $f$. Indeed, for certain  functions that are only piecewise continuous and have jump discontinuities, the variance grows  as $\sim \log n$ (which is of course still smaller than $\sim n$) see \cite{CL,SoshInf}. For determinantal point processes with growing variances general Central Limit Theorems formulated  \cite{HKPV,SoshCLT}. However, in these proofs the assumption of a growing variance is essential.
\end{remark}

An important example of a Central Limit Theorem was given in a remarkable work by Johansson \cite{Jduke}, where he proved \eqref{eq:RMTCLT} in the case of unitary ensembles in the one-cut situation. More precisely, he considered the orthogonal polynomial ensembles\footnote{In fact, \cite{Jduke} deals with unitary ensembles for general $\beta$. In the present paper we restrict ourselves to  the case $\beta=2$.} with  ${\rm d} \mu(x)={\rm d} \mu_n(x)= {{\rm e} ^{-n V(x)}}{\rm d} x$ and a polynomial $V$ of even degree satisfying certain assumptions. One of the crucial assumptions is that the support of the equilibrium measure $\mu_V$ (minimizing the logarithmic energy with external field $V$) is a single interval, say $[\gamma,\delta]$. For such ensembles Johansson \cite{Jduke} showed that for any real-valued sufficiently smooth function $f$ we have
\begin{align}\label{eq:joh}
X_f^{(n)} -\EE X_f^{(n)} \to N\left(0,\sum_{k \geq 1} k  |\hat f_k|^2 \right), \qquad \text{ as } n \to \infty,
\end{align}
in distribution, where   the coefficients $\hat f_k$ are defined as 
$$
 \hat f_k=\frac{1}{2\pi} \int_0^{2\pi}  f\left(\frac{\delta-\gamma}{2} \cos  \theta+\frac{\delta+\gamma}{2} \right) {\rm e}^{-{\rm i} k \theta} {\rm d}\theta,
$$
for $k\geq 1$.  The conditions on $V$ were further relaxed to allow more general analytic $V$ by Kriecherbauer and Shcherbina \cite{KS} by a refinement of the techniques in \cite{Jduke}. However, the one interval assumption for the support of the equilibrium measure can not be relaxed, since it is known that in the multi-cut case the fluctuations have a more exotic and not necessarily Gaussian  behavior, as recently analyzed by Shcherbina in  \cite{Shch}.

In this paper we revisit the problem  of finding  conditions that imply \eqref{eq:RMTCLT} using an approach that is quite different from \cite{Jduke,KS}. The starting point of our analysis is the assumption that the biorthogonal families are subject to  a recurrence relation. We  assume that there exists a \emph{banded} matrix $J$ such that 
\begin{equation}\label{eq:biorecur}
x \begin{pmatrix}  \psi_0(x)\\ \psi_1(x) \\ \psi_2(x) \\ \vdots\end{pmatrix}  =J \begin{pmatrix}  \psi_0(x)\\ \psi_1(x) \\ \psi_2(x) \\ \vdots\end{pmatrix}  ,
\end{equation}
for $x\in \bbR.$  Since $J$ is banded, the relation \eqref{eq:biorecur} indeed defines a finite term recurrence for the functions $\psi_j^{(n)}$ where the number of terms is bounded by the number of non-trivial diagonals of $J$. We call biorthogonal ensembles satisfying \eqref{eq:biorecur} \emph{biorthogonal ensembles with a recurrence}. It is a classical result that in the case of orthogonal polynomial ensembles  \eqref{eq:biorecur} always holds, with $J$ a tridiagonal matrix, known as the Jacobi matrix associated with $\mu$.  It is also known to hold for other ensembles based on polynomials, such as the class of Multiple Orthogonal Polynomial Ensembles \cite{Kuijlaars} and the two matrix model that will be discussed in more detail in Section~\ref{sec:two}.
 
In a nutshell, the main observation of the paper is that the fluctuations of any fixed moment of the empirical measure can be expressed in terms of a finite size submatrix of the recurrence matrix $J$ around the $(n,n)$-entry (where the size of the submatrix depends on the moment, cf. Lemma \ref{lem:keylemma}). If this submatrix has a limit along a certain subsequence, i.e. $J$ has a  \emph{right limit}, then the fluctuations also have a limit along the same subsequence. Using the right limit we provide a limiting expression for the fluctuation in that case. In general, the limiting expression is rather complicated, but in case the right limit is  a Laurent matrix it can be explicitly computed and we derive a Central Limit Theorem in the spirit of \eqref{eq:joh}.
 
The concept of right limit has been an extremely useful one in the spectral theory of Jacobi matrices. Notably, the right limits of a Jacobi matrix determine its essential spectrum \cite{S304} and also give information about the absolutely continuous spectrum \cite{LS, Remling}. Recently, an analogous concept for power series has also been shown to be useful for the understanding of when an analytic function has a natural boundary on its radius of convergence \cite{BS}. We find it remarkable that the notion of right limit also appears naturally in the context of the present paper. 
 
The approach that we follow  is new in that the conditions we formulate center on the asymptotic behavior of the recurrence coefficients and not on continuity properties of the underlying measure. As we shall discuss below, in the special case of orthogonal polynomial ensembles, this allows us to consider quite general measures, including some which are purely singular with respect to Lebesgue measure. Moreover, since a particular corollary to the  assumptions on $V$ in \cite{Jduke,KS} is that the recurrence coefficients of the orthogonal polynomials have a limit as $n\to \infty$, we reproduce the Central Limit Theorem in  \cite{Jduke,KS}.  Our results also apply to discrete orthogonal polynomial ensembles  related to tilings of large planar domains. As an example of the latter, we will discuss the Hahn ensemble and the corresponding lozenge tilings of hexagons in some detail. 
 
Finally, the structure of biorthogonal ensembles is known to be essentially more complicated (in contrast to the orthogonal polynomial ensembles, the underlying correlation kernel is not self-adjoint) and few general results can be found in the literature. To the best of our knowledge, the results in this paper are the first general Central Limit Theorems for these ensembles. To illustrate these results we will briefly discuss a case study of the two matrix model.

{\bf Acknowledgments} We thank Alexei Borodin, Torsten Ehrhardt, Kurt Johansson and Ofer Zeitouni for useful discussions.

\section{Statement of results}

We now state our main results. The proofs will be given in later sections.

\subsection{Biorthogonal ensembles with a recurrence}

Our main results are for biorthogonal ensembles \eqref{eq:defBO} with a recurrence given  by \eqref{eq:biorecur}. In many models of interest, such as the Unitary Ensembles, the parameters in the biorthogonal ensemble vary with $n$. Hence we want to also allow the measure $\mu$, and the functions $\psi_j$ and $\phi_j$ to depend on $n$, and write
$$
\mu=\mu_n, \qquad \psi_j= \psi_j^{(n)}, \qquad \phi=\phi^{(n)}_j \qquad  \text{and} \qquad  J= J^{(n)}.
$$
In this case we always assume that the band structure of $J^{(n)}$ (the number of non-trivial diagonals) does not depend on $n \in \bbN$. 

As discussed above, we want to describe conditions for a Central Limit Theorem \eqref{eq:RMTCLT} using the asymptotics of the recurrence coefficients as $n \rightarrow \infty$. The simplest situation is when the recurrence coefficients have a limit. In the general situation, we shall do this using the notion of right limit.

Let $\{\mathcal B^{(n)}\}_{n\in \bbN}$ be a sequence of banded infinite matrices  $\mathcal B^{(n)}=\left( \mathcal B^{(n)}_{r,s} \right)_{r,s=1}^\infty$. We say that $\mathcal B^R=(\mathcal B^R_{r,s})_{r,s=-\infty}^\infty$ is  a \emph{right} limit of $\{\mathcal B^{(n)}\}_{n\in \bbN}$  iff there exists a subsequence $\{n_j\}_{j}$ such that $$
\left({\mathcal B^R}\right)_{r,s}= \lim_{j\to \infty} \mathcal B_{n_j+r,n_j+s}^{(n_j)},$$
and $\mathcal B^R$ is bounded as an operator on $\ell_2(\bbZ)$. 

Put differently, a right limit is a subsequential limit \emph{in the strong operator topology} of shifts of the sequence of operators. If the sequence $\left\{ \mathcal B^{(n)} \right \}_n$ is uniformly bounded then a right limit exists by compactness. A right limit is not necessarily unique (at the very least, if $\mathcal B^R$ is a right limit, then any shift of $\mathcal B^R$ is a right limit as well). Note that we specifically demand that $\mathcal B^R$ is bounded, even if $\mathcal B^{(n)}$ is unbounded. 

In the first main result we assume that for some polynomial $p$ we have that $p(J^{(n)})$ has a right limit that has constant values on each diagonal. A two-sided infinite matrix with constant diagonals is called a Laurent matrix. Laurent matrices are completely determined by their symbol. Given a Laurent polynomial $s(w)=\sum_{j=-q}^p s_j w^j$, the Laurent matrix $L(s)$ is defined by 
$$ \left(L(s)\right)_{jk}=s_{j-k}, \qquad j,k \in \bbZ.
$$
We refer to $s$ as the \emph{symbol} of the Laurent matrix $L(s)$. 

Let $\{\mu_n, \psi_j^{(n)}, \phi_j^{(n)}, J^{(n)}\}$ be a sequence of biorthogonal ensembles with a recurrence. The following is our first main result.
\begin{theorem}\label{th:fluctuationsbio}
Let  $Q,P$ be two polynomials with real coefficients. Assume that $\{Q\left(J^{(n)}\right)\}_n$ has a right limit along a subsequence $\{n_j\}_j$  and assume the right-limit is the  Laurent matrix with symbol $s_Q(w)=\sum_{l=-q}^p s_l w^l$. Then 
$$
X_{P \circ Q}^{(n_j)} -\EE X_{P\circ Q}^{(n_j)} \to N\left(0,\sum_{k=1}^\infty k \widehat{(P\circ Q)}_k \widehat{(P\circ Q)}_{-k}\right), 
$$
where 
$$
\widehat{ (P\circ Q)}_k=\frac{1}{2 \pi {\rm i}} \oint_{|z|=1} P(s_Q(w)) w^k \frac{{\rm d}w}{w}.
$$
\end{theorem}

We postpone the proof of this theorem to Section \ref{sec:mainproof}.  We formulate separately the case of $Q=id$.

\begin{corollary}\label{cor:fluctuationsbio}
Assume that $\left\{J^{(n)}\right\}_n$ has a right limit along a subsequence $\{n_j\}_j$  and assume the right-limit is the  Laurent matrix with symbol $s(w)=\sum_{l=-q}^p s_l w^l$. Then for any polynomial $P$ with real coefficients  we have that 
$$
X_P^{(n_j)} -\EE X_P^{(n_j)} \to N\left(0,\sum_{k=1}^\infty k \hat P_k \hat P_{-k}\right), 
$$
where 
\beq \label{eq:cor:coef}
\hat P_k=\frac{1}{2 \pi {\rm i}} \oint_{|z|=1} P(s(w)) w^k \frac{{\rm d}w}{w}.
\eeq
\end{corollary}

\begin{remark}
The fact that $P$ has to have real coefficients is not a necessary condition. It is there to make sure that $X_P^{(n)}$ takes real values. The generalization to polynomials with complex coefficients is immediate and the only difference is that we now have a complex normal distribution. 
\end{remark}

\begin{remark}
The difference between Theorem \ref{th:fluctuationsbio} and Corollary \ref{cor:fluctuationsbio} is illustrated by the following example (also see Theorem \ref{th:ope-multi-interval} below). Suppose $\{J^{(n)}\}_{n\in \bbN}$ has a right limit $J^R$, but that the diagonals of $J^R$ are 2-periodic (in other words, the values alternate between two values). Then $J^R$ is not a Laurent matrix, but  $Q(J^R)$ is for some quadratic polynomial $Q$. Thus Corollary \ref{cor:fluctuationsbio} does not apply, but Theorem \ref{th:fluctuationsbio} does.
\end{remark}

\begin{remark} \label{rem:changeinJ}
We want to emphasize that there is a certain amount of freedom in the choice of $\psi_{j}$ and $\phi_{k}$, and hence also in $J$. For instance, for any given sequence of nonzero numbers $\{d_k\}_k$ the sequences $\tilde \psi_k= d_k \psi_k$ and $\tilde \phi_l=d_l^{-1} \psi_l$ are also biorthogonal and give rise to the same biorthogonal polynomial ensemble. However, the recurrence changes! Indeed, if $D$ is the diagonal matrix with $d_k$ along the diagonal, we now have $\tilde J= D  J D^{-1}$. Hence even in case $J$ has no right-limit, $\tilde J$ may have one.  To illustrate that there is no conflict with Theorem \ref{th:fluctuationsbio}, consider for example the case where $J$ has a right-limit that is a Laurent matrix with symbol $s$, and take $d_k=r^k$ for $r>0$. Then $\tilde J= D J D^{-1}$ has also right limit (along the same subsequence) which is also a Laurent matrix whose symbol now is $s_r(w)=s(rw)$.  It is not hard to show that $\hat P_k$ now gets an additional factor $r^{-k}$ but these cancel against each other in the computation of the variance. Thus, the value of $\hat P_k$ depends on the choice of $\psi_k$, $\phi_l$ and hence of $J$, but the variance does not. 

\end{remark}

\begin{remark}
It has often been observed that there is an interesting connection between Central Limit Theorems in random matrix theory and the Strong Szeg\H{o} Limit Theorem for the asymptotic behavior of determinants of Toeplitz matrices \cite{Jannals}. In fact, when dealing with the Circular Unitary Ensemble the two are identical. From this perspective, it is perhaps not surprising that this theorem and its relatives also play an important role in the methods of the present paper. In particular, the proof of Theorem \ref{th:fluctuationsbio} uses  Ehrhardt's generalization \cite{E} of the Helton-Howe-Pincus formula on a particular determinant involving Toeplitz operators.
\end{remark}

In Theorem \ref{th:fluctuationsbio} and Corollary \ref{cor:fluctuationsbio} we assumed the that the function in the linear statistic is a polynomial. When compared to \eqref{eq:joh}  one would like to extend Theorem \ref{th:fluctuationsbio} and Corollary \ref{cor:fluctuationsbio} so that they   hold for a more  general class of functions, say $f\in C^1(\bbR)$. The next theorem presents a criterion in the case that  $J^{(n)}$ is symmetric.  Note that if $L(s)$ is a right limit of  symmetric banded matrices, then $s(z)\in \bbR$ for $|z|=1$.

 \begin{theorem} \label{th:approx}
 Let $\phi_j^{(n)}=\psi_j^{(n)}$ (and hence $J^{(n)}$ is symmetric). Suppose that there exists a compact $E\subset \bbR$ such that for all $k\in \bbN$ we have 
\beq \label{eq:criterium}
\sum_{j=0}^{n-1}   \int_{\bbR \setminus E} x^k \phi_j^{(n)}(x)\psi_j^{(n)}(x) {\rm d} \mu_n(x) = o(1/n),
 \eeq
 as $n\to \infty$. Then Theorem \ref{th:fluctuationsbio} $($and so Corollary \ref{cor:fluctuationsbio}$)$ also hold with $P$ replaced by  any  real-valued $f \in C^1(\bbR)$ such that $|f(x)| \leq C (1+|x|)^k$ for some $C>0,k\in \bbN$.  
 \end{theorem}

The proof of the theorem follows by approximating the function $f$ by polynomials, which is why we need the compactness of $E$. Roughly, the criterion \eqref{eq:criterium} says that the point process essentially takes place on a compact set. If $\mu$ has compact support then this holds trivially. However, \eqref{eq:criterium} is also satisfied (with even a stronger order term) for unitary ensembles with a  real analytic potential that grows sufficiently fast at infinity. This can be checked for example by using the asymptotic results in \cite{DKMVZ1,DKMVZ2}. Thus, Theorem \ref{th:approx} extends the results of \cite{Jduke,KS}. Moreover, Theorem \ref{th:approx} also contains classical discrete orthogonal polynomial ensembles that have been of interest in tiling models of planar domain  \cite{BKMM,Jannals2}. As an example we will discuss its consequence to the Hahn ensemble in Subsection 6.2, which corresponds to lozenge tilings of a hexagon. 

\begin{remark}
An important issue in the non-symmetric case is that it is not obvious how one  defines \eqref{eq:cor:coef} for functions $f\in C^1(\bbR)$ since it is not  a priori clear that $s(z)$ is real for $|z|=1$. However, with additional information from the model at hand, it should be possible to extend Theorem \ref{th:approx} and its proof to  non-symmetric cases.
\end{remark}
In the results so far, we have assumed that the right limit is a Laurent matrix. However, we also have a  general result.  Let $J$ be a bounded operator not depending on $n$ and let $J^R$ be a right limit.  For $M\in \bbN$ we define the  truncation $J^R_M$ of $J^R$  by $$
 \left(J_M^R \right)_{kl} = \begin{cases}
 \left(J^R\right)_{kl}, & k,l=-M,\ldots,M\\
 0, & \text{ otherwise.}
 \end{cases}
$$
We also define the operator $P_-$ on $\ell_2(\bbZ)$ as  the projection onto the negative coefficients, i.e.\  
$$
\left(P_- x\right)_r=\left\{ \begin{array}{cc} x_r &  r<0 \\
 0 & \textrm{otherwise.} \end{array} \right. 
$$

\begin{theorem}\label{th:maintheorem}
Assume that $J^R$ is a right limit of $J$ with subsequence $\{n_j\}_j$.  Then for any  polynomial $f$, we have 
\begin{multline}\label{eq:maintheorem}
\lim_{j\to \infty} \EE \left[\exp z \left(  X_f^{(n_j)} -\EE X_f^{(n_j)} \right)\right]\\ = \lim_{M\to \infty}   {\rm e}^{-z \Tr P_{-} f(J^R_M)P_{-} }  \det  \left(I+P_{-} ({\rm e}^{z   f(J^R_M) } -I)P_{-}\right).
\end{multline}
uniformly for $z$ in a sufficiently small neighborhood of the origin. In particular, both limits exist. 
\end{theorem}

\begin{remark}
Note that although this theorem is only for non-varying $J$, we have no assumptions on the right-limit, giving a fairly general conclusion. Particularly in view of the the multi-cut results in \cite{P,Shch}, it is  interesting  to  further investigate of the right-hand side of \eqref{eq:maintheorem}, possibly under various assumptions on $J^R$. We will not pursue this in the present paper.  
\end{remark}

\subsection{Applications to orthogonal polynomial ensembles}

We next demonstrate the scope of the applicability of Theorem \ref{th:fluctuationsbio} by formulating a few corollaries for orthogonal polynomial ensembles with measures of compact support. 

Let ${\rm d}\mu(x)=w(x){\rm d}x+{\rm d}\mu_{\textrm{sing}}(x)$ be a probability measure on $\bbR$ with compact support, where $\mu_{\textrm{sing}}$ is the part of $\mu$ that is singular with respect to Lebesgue measure. Let $\{p_j\}_{j=0}^\infty$ be the orthogonal polynomials with respect to $\mu$, namely $\deg p_j=j$ and 
$$
\int_\bbR p_j(x) p_k(x) {\rm d} \mu(x) =\delta_{jk}, 
$$
for $j,k=0,1,\ldots.$  We now consider the orthogonal polynomial ensemble of size $n$, which  is the probability measure on $\bbR^n$ as defined in \eqref{eq:defBO} with the special choice $\phi_j=\psi_j=p_j$.  The recurrence \eqref{eq:biorecur} is in this case given by the well-known three term recurrence relation for orthogonal polynomials,
\begin{equation}\label{eq:oprecur}
\begin{split}
x p_n(x)&=a_{n+1} p_{n+1}(x)+b_{n+1} p_n(x)+a_{n} p_{n-1}(x), \quad n\geq 1, \\
xp_0(x)&=a_1p_1(x)+b_1p_0(x)
\end{split}
\end{equation}
where $a_n>0$ and $b_n \in \bbR$ are coefficients associated with $\mu$ \cite{deift}. Thus, orthogonal polynomial ensembles are clearly a special case of a biorthogonal ensemble with a recursion, where $J$ is a tridiagonal matrix. In this case, $J$ is called \emph{the Jacobi matrix associated with $\mu$}.

\subsubsection{The case of a single interval}

Here is an immediate corollary of Theorem \ref{th:approx}.

\begin{theorem}\label{th:theorem1}
Assume that there exists a subsequence $\{n_{j}\}_j$ such that  
 \beq \label{eq:nevaiclass}
a_{n_j}\to a, \quad b_{n_j} \to b
\eeq 
as $j\to \infty$ for some $a>0$, $b$. Then for any real-valued $f\in C^1(\bbR)$ we have
$$
X_f^{(n_j)} -\EE X_f^{(n_j)} \to N\left(0,\sum_{k \geq 1} k  |\hat f_k|^2 \right), \qquad \text{ as } j \to \infty,
$$
in distribution, where   the coefficients $\hat f_k$ are defined as 
$$ 
 \hat f_k=\frac{1}{2\pi} \int_0^{2\pi}  f(2 a \cos  \theta+b) {\rm e}^{-{\rm i} k \theta} {\rm d}\theta,
$$
for $k\geq 1$. 
\end{theorem}

A measure whose recurrence coefficients satisfy \eqref{eq:nevaiclass} is said to be in the Nevai class for the interval $[b-2a,b+2a]$.  Properties of this class of measures have been  intensively studied in the orthogonal polynomial literature (see \cite{SimonSzego} and references therein). In particular, if $\mu$ is in Nevai class for $[b-2a,b+2a]$ then $\sigma_{\textrm{ess}}(\mu)=[b-2a,b+2a]$, where $\sigma_{\textrm{ess}}(\mu)$   is the support of the measure $\mu$ with isolated points removed. 

To formulate a corollary that does not explicitly involve conditions on the recurrence coefficients, the following result follows immediately from Theorem \ref{th:theorem1} and the celebrated Denisov-Rakhmanov Theorem \cite[Theorem 1.4.2]{SimonSzego} (the original results comprising the Denisov-Rakhmanov Theorem are in \cite{Den, R1, R2}).

\begin{theorem} \label{th:ope-single-interval}
Suppose that $\sigma_{\textrm{ess}}(\mu)=[\gamma,\delta]$, an interval, and suppose further that $w(x)>0$ for Lebesgue a.e.\ $x \in [\gamma,\delta]$. Then for any real-valued $f\in C^1(\bbR)$ we have
$$
X_f^{(n)} -\EE X_f^{(n)} \to N\left(0,\sum_{k \geq 1} k  |\hat f_k|^2 \right), \qquad \text{ as } n \to \infty,
$$
in distribution, where   the coefficients $\hat f_k$ are defined as 
$$
 \hat f_k=\frac{1}{2\pi} \int_0^{2\pi}  f\left(\frac{\delta-\gamma}{2} \cos  \theta+\frac{\delta+\gamma}{2} \right) {\rm e}^{-{\rm i} k \theta} {\rm d}\theta,
$$
for $k\geq 1$.
\end{theorem}

\begin{proof}[Proof of Theorem \ref{th:ope-single-interval}]

Theorem 1.4.2 in \cite{SimonSzego} says (after scaling and shifting) that the conditions of Theorem \ref{th:ope-single-interval} imply that $\mu$ is in Nevai class for the interval $[\gamma,\delta]$, namely that the recursion coefficients satisfy $a_n \rightarrow \frac{\delta-\gamma}{4}$ and $b_n \rightarrow\frac{\delta+\gamma}{2}$ as $n \rightarrow \infty$. Thus, the result follows from Theorem \ref{th:theorem1}. 
\end{proof}

Note that the conditions of Theorem \ref{th:ope-single-interval} are phrased only in terms of properties of the underlying measure. In particular, the orthogonal polynomial ensemble for any absolutely continuous (w.r.t.\ Lebesgue) measure, such that the support of the weight is an interval, satisfies a Central Limit Theorem. No other assumption on the weight (e.g.\ continuity) is needed. 
Moreover, the Nevai class for $[\gamma,\delta]$ contains also many measures which are purely singular with respect to Lebesgue measure. One family of such examples is described in Section \ref{sec:OPEexamples} below. By Theorem \ref{th:theorem1}, the orthogonal polynomial ensemble associated with any such measure obeys a Central Limit Theorem.

It is important to note that $\sigma_{\textrm{ess}}(\mu)=[\gamma,\delta]$, or even $\supp(\mu)=[\gamma,\delta]$ are not sufficient conditions to ensure a Central Limit Theorem for the associated orthogonal polynomial ensemble. In Section \ref{sec:OPEexamples} we describe two examples with $\supp(\mu)=[\gamma,\delta]$ such that the associated orthogonal polynomial ensembles have different Gaussian limits along different subsequences. Example \ref{ex:example2}, in fact, has a continuum of such limits.

\subsubsection{The case of several intervals}

In the case that the support of a measure is a disjoint union of intervals, it is well known that the corresponding orthogonal polynomial ensemble does not necessarily satisfy a Central Limit Theorem in general \cite{P, Shch}. There are cases, however, where we can prove a Central Limit Theorem for the linear statistics of some particular functions. 

Let $S=[\gamma_1,\delta_1]\cup[\gamma_2,\delta_2]\cup\ldots\cup[\gamma_\ell, \delta_\ell]$ be a union of intervals. The following is part of Theorem 5.13.8 of \cite{SimonSzego}:

\begin{proposition} \label{prop:simonPer}
Let $p \in \{1,2,\ldots\}$ and let $\rho_S$ be the equilibrium measure of $S$ $($for the logarithmic potential$).$ Assume that for any $1 \leq j \leq \ell$, $\rho_S(\left(\gamma_j,\delta_j \right))$ is rational with $p \rho_S(\left(\gamma_j,\delta_j \right)) \in \mathbb{Z}$. 
Then there is a polynomial $\Delta$ of degree $p$, with real coefficients, such that 
\beq \no
\Delta^{-1}([-2,2])=S.
\eeq
\end{proposition}

The polynomial $\Delta$ is constructed roughly as follows. The set $S$ defines a two-sheeted Riemann surface by taking two copies of $\mathbb{C} \setminus S$ and gluing them along $S$. One considers all Herglotz functions that satisfy certain regularity conditions (called \emph{minimal} Herglotz functions in \cite{SimonSzego}). This family forms a torus in a natural way and each function corresponds to a one sided Jacobi matrix, as its $m$-function. The rationality of the equilibrium measure implies that each of these matrices is a periodic Jacobi matrix of period $p$. It turns out that the traces of the $p$'th transfer matrices for all these $p$-periodic Jacobi matrices are all equal. These traces are the polynomial $\Delta$. For the minimal $p$ satisfying the conditions of Proposition \ref{prop:simonPer}, we denote the polynomial $\Delta$, associated with such $S$, by $\Delta_S$.

\begin{theorem} \label{th:ope-multi-interval}
Suppose that ${\rm d}\mu(x)=w(x){\rm d}x+{\rm d}\mu_{\textrm{sing}}(x)$ is a probability measure on $\bbR$ with compact support such that 
\beq \no
\sigma_{\textrm{ess}}(\mu)=[\gamma_1,\delta_1]\cup[\gamma_2,\delta_2]\cup\ldots\cup[\gamma_\ell, \delta_\ell]\equiv S.
\eeq 
Assume that the equilibrium measure $\rho_S$ satisfies the conditions of Proposition \ref{prop:simonPer} with minimal integer $p$. Assume further that $w(x)>0$ for Lebesgue a.e.\ $x \in S$. Then for any real-valued $f\in C^1(\bbR)$ we have
$$
X_{f\circ \Delta_S}^{(n)} -\EE X_{f\circ \Delta_S}^{(n)} \to N\left(0,\sum_{k \geq 1} k  |\hat f_k^p|^2 \right), \qquad \text{ as } n \to \infty,
$$
in distribution, where   the coefficients $\hat f_k^p$ are defined as 
\beq \label{eq:coeffPer}
\hat f_k^p=\frac{1}{2\pi} \int_0^{2\pi}  f(2 \cos p  \theta) {\rm e}^{-{\rm i} k \theta} {\rm d}\theta
\eeq
\end{theorem}

\begin{proof}
Let $J$ be the Jacobi matrix corresponding to $\mu$. By \cite[Theorem 5.13.8]{SimonSzego} the condition on $\rho_S$ implies that $S$ is the spectrum of a two sided $p$-periodic Jacobi matrix, $J_0$. By Theorem 1.2 of \cite{dks}, the conditions 
\beq \no
\sigma_{\textrm{ess}}(\mu)=S
\eeq
and
\beq \no
w(x)>0 \quad \textrm{a.e.} x \in S
\eeq
imply that all right limits, $J^R$ of $J$ are in the isospectral torus of $J_0$. Namely, if $J^R$ is a right limit of $J$ then $J^R$ is $p$-periodic, with spectrum $S$, and by Theorem 3.1 of \cite{dks} 
\beq \label{eq:magicformula}
\Delta_S(J^R)=\mathbb{S}^p+\mathbb{S}^{-p}
\eeq
where $\mathbb{S}$ is the shift operator (this is what \cite{dks} call `the magic formula'). 

Since $\Delta_S$ is a polynomial, $\Delta_S(J)$ is a banded matrix and if  $J^R$ is a right limit of $J$ along the sequence $n_j$ then $\Delta_S\left(J^R \right)$ is a right limit of $\Delta_S(J)$ along the same sequence. Together with \eqref{eq:magicformula}, this implies that the \emph{single} right limit of $\Delta(J)$ is in fact $\mathbb{S}^p+\mathbb{S}^{-p}$, a Toeplitz matrix. Theorem \ref{th:approx} ($S$ is compact) finishes the proof. 
\end{proof}

\subsection{Overview of the rest of the paper}
The rest of the paper is organized as follows. Section 3 presents some preliminaries from operator theory and the theory of determinantal point processes. Section 4 sets the stage for the proofs of Theorems  \ref{th:fluctuationsbio}, \ref{th:approx} and \ref{th:maintheorem} by studying the asymptotics of a matrix version of the formula for the cumulants of a linear statistic. The proofs of  \ref{th:fluctuationsbio}, \ref{th:approx} and \ref{th:maintheorem} are given in Section 5, and Section 6 presents examples for applications of our results. Subsection 6.1 deals with orthogonal polynomial ensembles whose associated measure is compactly supported,  Subsection 6.2 contains a discussion on the discrete Hahn ensemble and the relation to lozenge tilings of a hexagon, whereas Subsection 6.3 presents the application of our results to the two matrix model. 

\section{Preliminaries}
\subsection{Preliminaries from operator theory}\label{sec:prel}
For a compact operator $A$ on a separable Hilbert space we denote  the singular values by $\sigma_j(A)$ (we recall that the singular values are the  square roots of the eigenvalues of the compact self-adjoint operator $A^*A$). We then define the trace norm $\|\cdot\|_1$ and Hilbert-Schmidt norm  $\|\cdot\|_2$ by 
$$
\|A\|_p^p =\sum_j \sigma_j(A)^p, \qquad p=1,2. 
$$
We also denote the operator norm by $\|\cdot\|_\infty$.  If $\|A\|_1 < \infty$ we say that $A$ is of trace class and if $\|A\|_2<\infty$ we say that $A$ is Hilbert-Schmidt.  

The following identities are standard (see for example \cite{Simon} for more details).
\begin{enumerate}
\item $|\Tr A| \leq \|A\|_1$
\item $\|A B\|_p\leq \|A\|_p \|B\|_\infty$ for $p=1,2$.
\item $\|A B\|_1\leq \|A\|_2\|B\|_2$
\end{enumerate}

We also need the notion of Fredholm determinant. If $A$ is a trace class operator we can define the Fredholm determinant 
$$\det (1+A) =\prod_j (1+\lambda_j),$$
where $\lambda _j$ are the eigenvalues of $A$. We then have the following identity
\beq\label{eq:Fredcont}| \det (1+A)-\det (1+B)|\leq \|A-B\|_1 \exp ( \|A||_1+ \|B\|_1+1),
\eeq
for any two trace class operators $A,B$ (cf. \cite[Th. 5.2]{GK}).

Finally, we recall an identity by Ehrhardt \cite{E} that we will use, which is a generalization of the Helton-Howe-Pincus formula. If $A,B$ are bounded operators such that the commutator $[A,B]=AB-BA$ is trace class then
\beq\label{eq:HH}
\det {\rm e}^{-A} {\rm e}^{A+B} {\rm e}^{-B} =\exp \left(-\frac{1}{2} \Tr [A,B]\right).
\eeq
Note that if $A$ or $B$ is trace class then, $\Tr [A,B] =0$, but for general bounded operators, the trace does not necessarily vanish. Part of that statement is that $ {\rm e}^{-A} {\rm e}^{A+B} {\rm e}^{-B}-I$ is trace class. In fact, in \cite{E} it was proved that 
\beq \label{eq:erh}
\| {\rm e}^{-A} {\rm e}^{A+B} {\rm e}^{-B}-I\|_1\leq \|[A,B]\|_1 \phi \left( \|A\|_\infty,\|B\|_\infty \right),
\eeq
for some function $\phi$.

\subsection{Preliminaries from the theory of determinantal point processes}
It is well-known \cite{Bor} that a biorthogonal ensemble as defined in \eqref{eq:defBO} is a determinantal point process with kernel $K_n(x,y)$ defined by 
\beq\label{eq:defKn}
K_n(x,y)= \sum_{j=0}^{n-1} \psi_j(x) \phi_j(y),
\eeq
By definition, this means that for any bounded measurable function $g$ we have 
\begin{multline} \label{eq:detproc}
\EE \prod_{j=1}^n (1+g(x_j)) = 1+ \sum_{k=1}^n  \frac{1}{k!}\underset{k \text{ times}}{\underbrace{\int \cdots \int }} \det \left(K_n(x_i,x_j)\right)_{i,j=1}^k \\ \times g(x_1) \cdots g(x_k)  {\rm d} \mu(x_1) \cdots {\rm d}\mu(x_k).
\end{multline}
For more details and background on general determinantal point processes we refer to \cite{BorDet,HKPV,J4,K,L,Sosh,Sosh2}.  For what follows it is important to note that the expansion at the right-hand side of \eqref{eq:detproc} is a Fredholm determinant and we can write 
$$
\EE \prod_{j=1}^n (1+g(x_j)) = \det (1+g K_n)_{\mathbb L_2(\mu)}, 
$$
where $gK$ stands for the integral operator on $\mathbb L_2(\mu)$ with kernel $g(x) K_n(x,y)$. Note that by taking $g= {\rm e}^{z f}-1$ we see that for any linear statistic we have 
\beq \label{eq:detproc2}
\EE [ \exp (z X^{(n)}_f)]  = \det \left(1+\left( {\rm e}^{z f}-1 \right) K_n \right)_{\mathbb L_2(\mu)}, 
\eeq
for any $z\in \bbC$. 

Moreover, we have the following useful identities
$$
\EE X^{(n)}_f= \int f(x)K_n(x,x)  {\rm d}\mu(x)
$$
and
$$
\Var X^{(n)}_f= \int   f(x)^2 K_n(x,x) {\rm d} \mu(x)- \iint f(x) f(y) K_n(x,y) K_n(y,x) {\rm d}\mu(x) {\rm d} \mu(y).
$$
By the reproducing property, ($\int K_n(x,y)K_n(y,x) {\rm d} \mu(y) =K_n(x,x)$), we can rewrite the latter in the following way
\beq \label{eq:variancedeterminantal}
\Var X^{(n)}_f=  \frac{1}{2}\iint (f(x)- f(y))^2 K_n(x,y)K_n(y,x) {\rm d}\mu(x) {\rm d} \mu(y).
\eeq
Both the mean and the variance can also be written in operator form
\beq \no
\begin{split}
\EE X^{(n)}_f&= \Tr f K_n,\\
\Var X^{(n)}_f & =- \Tr [f,K_n] ^2,
\end{split}
\eeq
where $[f,K_n]$ is the commutator of $f$ (as a multiplication operator) and $K_n$.

\subsection{Cumulants}\label{sec:cumul}

The \emph{cumulants} $\mathcal C_m^{(n)}(f)$ for the linear statistic $X_f^{(n)}$ are defined by 
$$
\log \EE \left[\exp  z X_f^{(n)} \right]= \sum_{m=1}^\infty \mathcal  C_m^{(n)}(f) z^m.
$$
Note that by definition,  the first two cumulants are $\mathcal C_1^{(n)}(f)=\EE X^{(n)}_f$ and $\mathcal C_2^{(n)}(f)=\frac12\Var X_f^{(n)}$. Moreover, it is clear that the cumulants determine the moments for the linear statistic entirely. Hence, if one computes the limiting behavior of the cumulants we also obtain the limiting behavior of the moments.

By \eqref{eq:detproc2} we also have 
\begin{multline} \no
\log \EE \left[\exp  z X_f^{(n)} \right]=\log \det (1+({\rm e}^{z f}-1) K_n)_{\mathbb L_2(\mu)}\\=
\exp \Tr  (({\rm e}^{z f}-1) K_n)_{\mathbb L_2(\mu)}.
\end{multline}
We trust there is no confusion after dropping  the lower index ${\mathbb L_2(\mu)}$ which we will do from now on. The latter expression will provide us with a standard expression for the cumulants in terms of the kernel $K_n$. Indeed, by expanding $\log ( 1+ ({\rm e}^{z  f}-1)K_n)$ in powers of $z$ and changing the order of summation we have
\begin{multline}\no
\log \EE \left[\exp  z X_f^{(n)} \right] = 
\sum_{j=1}^\infty \frac{(-1)^j}{j} \Tr\left( ({\rm e}^{z  f} -1)K_n\right)^j\\
= \sum_{j=1}^\infty \frac{(-1)^j}{j} \sum_{{ l_1\geq 1,\ldots, l_j\geq 1}} z^{l_1+\cdots + l_j}\frac{\Tr  f^{l_1} K_n\cdots  f^{l_j} K_n}{l_1! \cdots l_j!}\\
= \sum_{j=1}^\infty \frac{(-1)^j}{j} \sum_{m=j}^\infty z ^m \sum_{\overset{l_1+ l_2+ \cdots l_j=m}{ l_i\geq 1}}\frac{\Tr f^{l_1} K_n\cdots   f^{l_j} K_n}{l_1! \cdots l_j!}\\
= \sum_{m=1}^\infty z ^m  \sum_{j=1}^m\frac{(-1)^j}{j}  \sum_{\overset{l_1+ l_2+ \cdots l_j=m}{ l_i\geq 1}}\frac{\Tr f^{l_1} K_n\cdots   f^{l_j} K_n}{l_1! \cdots l_j!}.
\end{multline}
for sufficiently small $z$. Hence
\beq\label{eq:cumulantexpression0}
\mathcal C^{(n)}_m(f) = \sum_{j=1}^m \frac{(-1)^j}{j} \sum_{\overset{l_1+ l_2+ \cdots l_j=m}{ l_i\geq 1}} \frac{\Tr f^{l_1} K_n\cdots 
 f^{l_j} K_n}{l_1! \cdots l_j!}.
\eeq
This expression for the cumulants will be essential in our analysis.
\subsection{Some results on Toeplitz operators}

We end this section with recalling some definitions and  results from the theory of Toeplitz matrices/operators. For a general reference we refer to~\cite{BStoeplitz}.

Let $s(w)=\sum_{j=-q}^p s_j w^j$ be a Laurent polynomial for some $p,q\in \bbN$. Then we can associate with $s$ the Laurent operator $L(s)$, the Toeplitz operator $T(s)$ and the Hankel operator $H(s)$, by the infinite matrices
\begin{align}\no
\left( L(s) \right)_{j,k}=s_{j-k},& \qquad j,k\in \bbZ\\
\left( T(s) \right)_{j,k}=s_{j-k},& \qquad  j,k=1,2,\ldots\\
\left( H(s) \right)_{j,k}=s_{j+k-1},& \qquad  j,k=1,2,\ldots
\end{align}
Of course the definitions extend to more general symbols $s$, but in our situations the symbol $s$ will always be a Laurent polynomial so we restrict ourselves to that case. Note that  for Laurent polynomials we have that $L(s)$ and $T(s)$ are banded and $H(s)$ is of finite rank.

For any two symbols $s_1,s_2$ we have
\begin{align} \no
L(s_1s_2)&=L(s_1)L(s_2)\\
T(s_1s_2) &=T(s_1) T(s_2) +H(s_1) H(\tilde s_2),
\end{align}
where $\tilde s_2(z)=s_2(1/z)$. The second identity has two important   corollaries. First, 
\beq \label{eq:comtopli}
[T(s_1),T(s_2)]= -H(s_1) H(\tilde s_2)+H(s_2) H(\tilde s_1)\eeq
Second, if we split a given symbol by $s=s_++s_-$ with $s_+(w) =\sum_{j\geq 0} s_j w^j$, the projection onto the analytic part, then 
$$T(s_\pm ) T(s_\pm)= T(s_\pm^2)$$
which implies that for exponentiated symbols  
\beq \no
T({\rm e}^{s_\pm })= {\rm e}^{T(s_\pm)}.
\eeq
By combining this with  \eqref{eq:HH} we find the following identity
\begin{multline}\no
\det T({\rm e}^{-s_+}){\rm e}^{T(s)} T({\rm e}^{-s_-})= \det {\rm e}^{-T(s_+)}{\rm e}^{T(s)}  {\rm e}^{-T(s_-)}\\=\exp  \left(-\frac12\Tr [T(s_+),T( s_-)]\right)=\exp \frac12  \Tr H(s_+) H(\widetilde{( s_-)})\\=\exp\left(\frac{1}{2}\sum_{k\geq 1} k s_k s_{-k} \right),
\end{multline}
see also Example 3.1 in \cite{E}. 

The following lemma is of particular importance to us. To the best of our knowledge, the result has not appeared before in the literature. The proof  is a based on \eqref{eq:erh} and a trick that was also used in \cite{BW} to give a proof of the Borodin-Okounkov identity for Toeplitz determinants.    We define $P_n$ as the projection in $\ell_2(\bbN)$ onto the first $n$ components.  
\begin{lemma}\label{lem:toeplitz}
Let $K$ be a trace class operator on $\ell_2(\bbN)$, $s(w)$ a Laurent polynomial and $T(s)$ the corresponding Toeplitz operator. Then 
\beq
\lim_{n\to \infty} {\rm e}^{-\Tr P_n (T(s)+K)} \det (I+P_n({\rm e}^{T(s)+K}-I))=\exp\left(\frac{1}{2}\sum_{k\geq 1} k s_k s_{-k} \right). \no
\eeq
\end{lemma}
\begin{proof}
Let us first deal with the case $K=0$.  In that case, split $s=s_++s_-$ where $s_+$ is the analytic part of $s$, and write  $T(s)=T(s_+)+T(s_-)$. Note that $T(s_+)$ is lower triangular (including the diagonal) and $T(s_-)$ upper triangular. By this triangular structure we have
$$P_n T(s_+)^2= P_n T(s_+) P_n T(s_+) P_n,  $$
and by iterating this identity we get
$$ {\rm e}^{-P_n T(s_+)  P_n }= (1-P_n) + P_n {\rm e}^{- T(s_+)}.$$
Similarly,
 by iterating 
$$ T(s_-)^2 P_n= P_n T(s_-) P_n T(s_-) P_n,  $$
we also have 
$$ {\rm e}^{- T(s_-) P_n }= (1-P_n) + {\rm e}^{- T(s_-)}  P_n.$$ Moreover, 
$${\rm e}^{-\Tr P_n T(s)  P_n } = {\rm e}^{-\Tr P_n T(s_+) P_n } = \det ({\rm e}^{-P_n T(s_+)P_n}) =\det\left((1-P_n) +  P_n {\rm e}^{- T(s_+)}  \right),$$ 
and $$ 1= {\rm e}^{-\Tr P_n T(s_-) P_n } = \det ({\rm e}^{-P_n T(s_-)P_n}) =\det\left((1-P_n) +   {\rm e}^{- T(s_-)} P_n  \right).$$
By taking the product matrices we obtain  
\begin{multline}\no
\det\left((1-P_n) + P_n{\rm e}^{- T(s_+)}  \right) \det\left((1-P_n) + P_n  {\rm e}^{T(s)} P_n\right) \det\left((1-P_n) +   {\rm e}^{- T(s_-)} P_n \right)\\= \det \left((1-P_n) + P_n  {\rm e}^{- T(s_+)} {\rm e}^{ T(s)} {\rm e}^{- T(s_-)} P_n \right)\\= \det \left(I+ P_n \left( {\rm e}^{- T(s_+)} {\rm e}^{ T(s)} {\rm e}^{- T(s_-)}-I \right) P_n \right).
\end{multline}
Here we also used that $(1-P_n)P_n=0$ and that $P_n{\rm e}^{- T(s_+)} = P_n{\rm e}^{- T(s_+)} P_n$ and ${\rm e}^{- T(s_-)} P_n=P_n {\rm e}^{- T(s_-)} P_n$. Hence we get  
\beq\label{eq:triangdet}
{\rm e}^{-\Tr P_n T(s)} \det (I+P_n ({\rm e}^{T(s)}-I)P_n) = \det \left(I+ P_n \left( {\rm e}^{- T(s_+)} {\rm e}^{ T(s)} {\rm e}^{- T(s_-)}-I \right) P_n \right).\eeq
Note that by using \eqref{eq:comtopli}, the fact that Hankel operators of Laurent symbols have finite rank and the fact that the trace class operators are an ideal in the space of bounded operators, we find that $[T(s_+),T(s_-)]$ is trace class.  Hence by \eqref{eq:erh} we also have that $${\rm e}^{- T(s_+)} {\rm e}^{ T(s)} {\rm e}^{- T(s_-)}-I$$ is of trace class. From here it is standard (see for example \cite[Th. 5.5]{GK}) to show that 
$$ P_n \left( {\rm e}^{- T(s_+)} {\rm e}^{ T(s)} {\rm e}^{- T(s_-)}-I \right) P_n \to  {\rm e}^{- T(s_+)} {\rm e}^{ T(s)} {\rm e}^{- T(s_)}-I,$$
in trace norm, as $n\to \infty$.  Moreover, by continuity of the Fredholm determinant \eqref{eq:Fredcont} and \eqref{eq:HH} we have
\begin{multline}
\lim_{n \to \infty} \det \left(I+ P_n \left( {\rm e}^{- T(s_+)} {\rm e}^{ T(s)} {\rm e}^{- T(s_-)}-I \right) P_n \right) \\= \det   {\rm e}^{- T(s_+)} {\rm e}^{ T(s)} {\rm e}^{- T(s_-)}= \exp \left(- \frac{1}{2} \Tr [T(s_+),T(s_-)]\right)\end{multline} 
By following \eqref{eq:comtopli} we have  $\Tr [T(s_+),T(s_-)]=-\Tr  [H(s_+),H(\widetilde{ (s_-)})]$ and by computing the trace  we arrive at the result.

Finally, let us consider the case $K\neq 0$.  By using the same strategy that lead to  \eqref{eq:triangdet},  we now end up with 
\begin{multline*}
{\rm e}^{-\Tr P_n T(s)} \det (I+P_n ({\rm e}^{T(s)+K}-I)P_n) \\= \det \left(I+ P_n \left( {\rm e}^{- T(s_+)} {\rm e}^{ T(s)+K} {\rm e}^{- T(s_-)}-I \right) P_n \right)
\end{multline*}
We now use that $[T(s_+),T(s_-)]$ and $ {\rm e}^{ T(s)+K}- {\rm e}^{ T(s)}$ are both trace class (the latter follows by expanding the exponentials),  to conclude that
\begin{multline*}
\left( {\rm e}^{- T(s_+)} {\rm e}^{ T(s)+K} {\rm e}^{- T(s_-)}-I \right)\\=\left( {\rm e}^{- T(s_+)} {\rm e}^{ T(s)} {\rm e}^{- T(s_-)}-I \right)+{\rm e}^{- T(s_+)}\left( {\rm e}^{ T(s)+K}- {\rm e}^{ T(s)}\right) {\rm e}^{- T(s_-)}
\end{multline*}
is trace class. Hence we have 
\beq\no
\lim_{n\to \infty} {\rm e}^{-\Tr P_n T(s)} \det (I+P_n ({\rm e}^{T(s)+K}-I)P_n) = \det {\rm e}^{-T(s_+)} {\rm e}^{T(s)+K} {\rm e}^{-T(s_-)}.
\eeq
By using $P_n K \to K$ in trace norm and $\det {\rm e}^K= \exp(-\Tr K)$ we thus have
\begin{multline*}
\lim_{n\to \infty} {\rm e}^{-\Tr P_n (T(s)+K)} \det (I+P_n ({\rm e}^{T(s)+K}-I)P_n) \\
= \det {\rm e}^{-T(s_+)} {\rm e}^{T(s)+K} {\rm e}^{-T(s_-)}{\rm e}^{- K}\\= \det {\rm e}^{-T(s_+)} {\rm e}^{T(s)+K}  {\rm e}^{-T(s_-)-K} {\rm e}^{-T(s_-)+K}{\rm e}^{-T(s_-)}{\rm e}^{- K}\\=\det {\rm e}^{-T(s_+)} {\rm e}^{T(s)+K}  {\rm e}^{-T(s_-)-K} \det {\rm e}^{-T(s_-)+K}{\rm e}^{-T(s_-)}{\rm e}^{- K},
\end{multline*}
 (see also the generalization of \eqref{eq:erh} given in \cite[Cor. 2.3]{E}).
The statement now follows by \eqref{eq:erh} and the fact that $\Tr [B,C]=0$ if one of the operators $B$ and $C$ is trace class. 
\end{proof}

\begin{corollary}\label{cor:toeplitzcumu}
Let $K$ be a trace class operator on $\ell_2(\bbN)$, $s(w)$ a Laurent polynomial and $T(s)$ the corresponding Toeplitz operator. Then, with $C^{(n)}_m \left(\mathcal B \right)$ as defined in \eqref{eq:newcumu} below,
\begin{align}
\lim_{n\to \infty} C_m^{(n)}(T(s)+K)=\begin{cases} 
\frac12\sum_{k=1}^\infty k s_k s_{-k}, & m=2\\
0, & \text{otherwise}.
\end{cases}
\end{align}
\end{corollary}
\begin{proof}
Since 
\[\log \det\left (I+P_n({\rm e}^{z(T(s)+K)}-I)\right) = \sum_{m=1}^\infty C_m^{(n)}(T(s)+K) z^m,\]
the statement follows from Lemma \ref{lem:toeplitz} (with $s$ and $K$ multiplied by $z$). 
\end{proof}

\begin{remark}
Lemma \ref{lem:toeplitz} and Corollary \ref{cor:toeplitzcumu} are both stated including a trace class operator $K$. However, in our proofs we will only use these results with $K=0$. Nevertheless, we wish to point out that the general result can be used for a more direct proof of Theorem \ref{th:fluctuationsbio} in case the recurrence matrix $J$ is a trace class perturbation of a Toeplitz operator. If the latter holds, Theorem \ref{th:fluctuationsbio} follows directly from Corollary \ref{cor:toeplitzcumu} and the first lines of the proof of Theorem \ref{th:fluctuationsbio} given in Section \ref{sec:proof}.
\end{remark}
\section{Some results for banded matrices}\label{sec:mainproof}

In this section, we will prepare the proofs of our main results, by studying properties of the cumulants \eqref{eq:cumulantexpression0} with $f$ replaced by a banded matrix,  $K_n$ replaced by the projection on the first $n$ components and $\mathbb L_2(\mu)$ by $\ell_2(\bbN)$. 

\subsection{Definition  of $C^{(n)}_m(\mathcal B)$ and its properties}

For any infinite matrix $\mathcal B=\left(\mathcal B_{rs}\right)_{r,s=1}^\infty$ we define $C_m^{(n)} (\mathcal B)$ by
\beq \label{eq:cumuB}
 C_m^{(n)}(\mathcal B) = \sum_{j=1}^m \frac{(-1)^j}{j} \sum_{\overset{l_1+ l_2+ \cdots l_j=m}{ l_i\geq 1}} \frac{\Tr \mathcal B^{l_1} P_n\cdots \mathcal B^{l_j} P_n }{l_1! \cdots l_j!}.
 \eeq
 Here $P_n$ is the projection in $\ell_2(\bbN)$ onto the first $n$ components.  We stress that we will be mainly interested in situations where $\mathcal B$ is banded. In that case, we may  (and do) allow $\mathcal B$ to be unbounded as an operator on $\ell_2(\bbN)$, as the band structure ensures that  there is no problem  when considering powers $\mathcal B^l$ for $l \in \bbN$. 

We start with the following simple, but crucial, observation.

\begin{lemma} 
Let $n \in \bbN$ and $m\geq 2$. Let $\mathcal B=(\mathcal B_{rs})_{r,s=1}^\infty$ be any banded infinite matrix.  Then we have
\beq  \label{eq:newcumu}
 C_m^{(n)}(\mathcal B) = \sum_{j=2}^m \frac{(-1)^j}{j} \sum_{\overset{l_1+ l_2+ \cdots l_j=m}{ l_i\geq 1}} \frac{\Tr \mathcal B^{l_1} P_n\cdots \mathcal B^{l_j} P_n-\Tr \mathcal B^m P_n }{l_1! \cdots l_j!}. 
\eeq
\end{lemma}
\begin{proof} 
Note that   by  expanding the right-hand side of $z=\log \left(1+(\exp z-1)\right)$ in  a power series in $z$, we have 
\beq\no
 \sum_{j=1}^m \frac{(-1)^j}{j} \sum_{\overset{l_1+ l_2+ \cdots l_j=m}{ l_i\geq 1}} \frac{1}{l_1! \cdots l_j!}=0, \quad m\geq 2.
\eeq
Hence \beq\no
 \sum_{j=1}^m \frac{(-1)^j}{j} \sum_{\overset{l_1+ l_2+ \cdots l_j=m}{ l_i\geq 1}} \frac{\Tr \mathcal B^m P_n }{l_1! \cdots l_j!}=0, \quad m\geq 2,
\eeq 
and thus \eqref{eq:newcumu} follows by subtracting the latter from \eqref{eq:cumuB}.
\end{proof}

This lemma brings us to the following key observation of the paper.

\begin{lemma} \label{lem:keylemma} Let $\mathcal B$ be  banded and bounded matrix and  let $b$ be   such that   $\mathcal B_{rs}=0$ if $r-s >b$. Then  for $M,n\in \bbN$ we have that 
\beq \no
\frac{\partial}{\partial \mathcal B_{rs}}  C_m^{(n)}(\mathcal B)=0, \begin{array}{l}{\text{ if $|s-n|> b M$}}\\{\text{ or  $|r-n|>b M$ }}
\end{array}
\eeq
for $m=1,\ldots, M$.
\end{lemma}
\begin{proof}
We start by writing each term in the definition of $C_m^{(n)}( \mathcal B)$ as 
\beq\no
\begin{split}
& \Tr P_n\mathcal B^{l_1} P_n\cdots \mathcal B^{l_j} P_n-\Tr  P_n\mathcal B^m P_n\\
&\quad=-\sum_{s=1}^n \sum_{\{r_1,\ldots, r_j\}\in I_{j,n}} \left( \mathcal B^{l_1} \right)_{sr_1} \left( \mathcal B^{l_2} \right)_{r_1r_2}\cdots\left( \mathcal B^{l_j} \right)_{r_{j-1}s},
\end{split}
\eeq
where $I_n$ is the index set 
$$I_{j,n} \in  \{  (r_1,\ldots,r_{j-1}) \in \bbN^{j-1} \mid \exists k :  r_k > n\}.$$  
Since $\mathcal B$ is banded we can further restrict the indices  to the set $I_{j,n}^*$ defined by
$$I_{j,n}^* \in  \{  (r_1,\ldots,r_{j-1}) \in \bbN^{j-1} \mid \exists k :  r_k > n , \quad \forall l |r_l- r_{l-1}|\leq b\}.$$
But the two defining relation for this set show that the only non-trivial contributions come from terms for which all indices $r_l \geq n-b (j-2) $. On the other hand, from the band structure and the fact that $s\leq n$ we also see that the only non-trivial contributions come from terms for which all indices $r_l \leq s+b (j-1)\leq n+b(j-1)$. Since $j$ runs from $1$ to $m$ and $m$ from $1$ to $M$,  the statement follows.
\end{proof}

\subsection{Asymptotic behavior of $C^{(n)}_m(\mathcal B^{(n)}_m)$}

We will now consider a sequence  $\{\mathcal B^{(n)}\}_{n\in \bbN}$ of banded infinite matrices  $\mathcal B^{(n)}=(\mathcal B^{(n)}_{r,s})_{r,s=1}^\infty$ where the number of non-trivial diagonals is independent of $n$. That is, we assume that there exists a $b>0$ such that $  B^{(n)}_{r,s}=0$ for $r-s>b$ and $n\in \bbN$.   Our goal is to analyze $C^{(n)}_m(\mathcal B^{(n)})$ as $n\to \infty$. 

We start with a definition.  For $m \in \bbN$ and $F$ a finite rank operator on $\ell_2(\bbZ)$ we define
\beq \label{eq:defDmF}
D_m(F)=\sum_{j=1}^m \frac{(-1)^j}{j} \sum_{\overset{l_1+\cdots +l_j=m}{l_i\geq 1}} \frac{\Tr \prod_{s=1}^j P_- F^{l_s} P_-  -P_- F^m P_-   }{l_1! \cdots l_j!}, \eeq
where we recall that $P_-$ is the projection onto the negative coefficients, i.e.  $\left(P_- x\right)_r=x_r$ for $r<0$ and $0$ otherwise.

Let us also recall that we say that $\mathcal B^R=({\mathcal B^R_{r,s}})_{r,s=-\infty}^\infty$ is  a right limit of $\{\mathcal B^{(n)}\}_{n\in \bbN}$  iff there exists a subsequence $\{n_j\}_{j}$ such that $$
\left({\mathcal B^R}\right)_{r,s}= \lim_{j\to \infty} \mathcal B_{n_j+r,n_j+s}^{(n_j)},$$
and $\mathcal B^R$ is bounded as an operator on $\ell_2(\bbZ)$. We also use the notation 
$\mathcal B^R_M$  for the matrix defined by 
$$\left({\mathcal B^R_M}\right)_{r,s}=\begin{cases} \left({\mathcal B^R}\right)_{r,s}, & \text{if } r,s=-M,\ldots,M\\
0,& \text{ otherwise.}
\end{cases}
$$
We then have the following result on the asymptotic behavior of $C^{(n)}_m(\mathcal B^{(n)})$.

\begin{lemma} \label{lem:fromctoD} 
Let $\{\mathcal B^{(n)}\}_{n\in \bbN}$ be a sequence of a banded infinite matrices and let $b$ be   such that   $\mathcal B_{rs}^{(n)}=0$ if $r-s >b$, for all $n\in \bbN$.  Let $\mathcal B^R$ be a right-limit of $\{\mathcal B^{(n)}\}_{n\in \bbN}$ along a subsequence $\{n_j\}_j$.
Let $M\in \bbN$. Then
\beq \label{eq:fromctoD}
\lim_{j\to \infty}C_m^{(n_j)}(\mathcal B^{(n_j)})= D_m( \mathcal B^R_{b M}),
\eeq
for $m=1,\ldots M.$
\end{lemma}
\begin{proof} 
Note that  $C^{(n)}_m(\mathcal B^{(n)})$ is defined for $\mathcal B^{(n)}$ acting on $\ell_2(\bbN)$ and $D_m(\mathcal B^R_{bM})$ for $\mathcal B_{bM}^R$ acting on $\ell_2(\bbZ)$. To prepare the limit \eqref{eq:fromctoD} we  embed $\ell_2(\bbN)$ into $\ell_2(\bbZ)$ and extend the one-sided infinite matrices $\mathcal B^{(n)}$ to  two-sided infinite matrices by setting  all the other entries to zero. Moreover, we extend  $P_n$ in the definition of $C^{(n)}_m$ to an operator on $\ell_2(\bbZ)$, by setting $\left(P_n x\right)_r=x_r$ for $r\leq n$ and $0$ otherwise.  One easily checks that this does not alter the value of $C^{(n)}_m$.

By Lemma \ref{lem:keylemma} we have that $C_m^{(n)}(\mathcal B^{(n)})$ depends only on the $\mathcal B_{rs}^{(n)}$ with $|n-r|, |n-s|\leq b M$. Hence we can set all the other entries of $\mathcal B$ to be identically zero.  In other words, we can replace $\mathcal B^{(n)}$ by the matrix
$ S_{n} \mathcal B^{R,n}_{b M} S_{n}^*$, where $$\left(\mathcal B^{R,n}_{b M}\right)_{rs}=\begin{cases}\mathcal B^{(n)}_{n+r,n+s},& \text{ for } r,s=-b M,\ldots b M.\\
0, & \text{otherwise.}
\end{cases} $$
and the operator $S_n$ is the shift operator $\left(S_nx\right)_s=x_{n+s}$, for $s\in \bbZ$. 
Then by also using $S^*_n P_n S_n=P_-$, it is easy to see that for $n$ large enough we have
$$C_m^{(n)}(\mathcal B^{(n)})=C_m^{(n)}(S_{n} \mathcal B^{R,n}_{bM} S_{n}^*)=D_m(f,\mathcal B_{bM}^{R,n}),$$
for $m=1,\ldots, M$. By setting $n=n_j$ and taking the limit $j \to \infty$ we obtain the statement.
\end{proof}

The following is a straightforward consequence of the Lemma \ref{lem:fromctoD}, that will be important to us later on. 

\begin{lemma}\label{lem:equivalent}
Let $\{\mathcal B^{(n)}_1\}_{n\in \bbN} $ and $\{\mathcal B^{(n)}_2\}_{n\in \bbN} $ be two sequences of banded infinite matrices that have the same right-limit $\mathcal B^R$ along the (same) subsequence $n_j$. Then for  any $m \in \bbN$ we have 
\begin{align} \no
\lim_{j\to \infty} C^{(n_j)}_m\left(\mathcal B_1^{(n_j)}\right)=\lim_{j\to \infty}  C^{(n_j)}_m\left(\mathcal B_2^{(n_j)}\right).
\end{align} 
  In particular, both limits exist.
\end{lemma}

\begin{proof}
The statement is  an easy  consequence of the fact that by Lemma \ref{lem:fromctoD} we have   for both $t=1,2$ that
$$\lim_{j\to \infty} C^{(n_j)}_m\left(\mathcal B_t^{(n_j)}\right) = D_m\left(\left(\mathcal B^R\right)_M\right),$$
for any $M$ sufficiently large.  \end{proof}

\subsection{The case of a bounded matrix $\mathcal B$}

In the previous paragraphs we allowed the matrix $\mathcal B$ (or $\mathcal B^{(n)}$) to be unbounded as an operator on $\ell_2(\bbN)$. We will now show that if $\mathcal B$ defines a bounded operator (and for simplicity also fixed) then we can organize the $\mathcal C_m^{(n)}(\mathcal B)$ in a Fredholm determinant that serves as a generating function. Moreover,  we obtain a limit for Fredholm determinants, from which Theorem \ref{th:maintheorem} is almost immediate.   
\begin{lemma} \label{lem:rewritingoffredholmdet}
Let $\mathcal B=({\mathcal B}_{r,s})_{r,s=1}^\infty$ be a bounded operator on $\ell_2(\bbN)$. For sufficiently small $z$ we have 
\beq \label{eq:cumull}
\exp(-z\Tr \mathcal B P_n) \det\left( 1+ ({\rm e}^{z  \mathcal B}-1)P_n\right)
= \exp\left(\sum_{m=2}^\infty z^m C_m^{(n)}(\mathcal B)\right)
\eeq 
where $C_m^{(n)}(\mathcal B)$ is defined by \eqref{eq:newcumu}.
\end{lemma}

\begin{proof}
The proof is similar to the proof of \eqref{eq:cumulantexpression0}, but now with $f$ and $K_n$ replaced by $\mathcal B$ and $P_n$.  We also need one additional step.

Again, by using $ \log \det  (1+A)= \Tr \log (1+A)$ and an expansion of $\log ( 1+ ({\rm e}^{z  \mathcal B}-1)P_n)$ in powers of $z$ we can rewrite the right-hand side of \eqref{eq:maintheoremhelp} as
\begin{multline}\no
\det\left( 1+ ({\rm e}^{z  \mathcal B}-1)P_n\right)= \exp \sum_{j=1}^\infty \frac{(-1)^j}{j} \Tr\left( ({\rm e}^{z  \mathcal B} -1)P_n\right)^j\\
=\exp \sum_{j=1}^\infty \frac{(-1)^j}{j} \sum_{{ l_1\geq 1,\ldots, l_j\geq 1}} z^{l_1+\cdots + l_j}\frac{\Tr  \mathcal B^{l_1} P_n\cdots  \mathcal B^{l_j} P_n}{l_1! \cdots l_j!}\\
=\exp \sum_{j=1}^\infty \frac{(-1)^j}{j} \sum_{m=j}^\infty z ^m \sum_{\overset{l_1+ l_2+ \cdots l_j=m}{ l_i\geq 1}}\frac{\Tr \mathcal B^{l_1} P_n\cdots  \mathcal B^{l_j} P_n}{l_1! \cdots l_j!}.
\end{multline}
for sufficiently small $z$. By changing the order of summation we obtain
$$
\det\left( 1+ ({\rm e}^{z  \mathcal B}-1)P_n\right)
= \exp\left(\sum_{m=1}^\infty z^m C_m^{(n)}(\mathcal B)\right)
$$
for sufficiently small $z$,
where 
\beq\label{eq:cumulantexpression}
C^{(n)}_m(\mathcal B) = \sum_{j=1}^m \frac{(-1)^j}{j} \sum_{\overset{l_1+ l_2+ \cdots l_j=m}{ l_i\geq 1}} \frac{\Tr \mathcal B^{l_1} P_n\cdots \mathcal B^{l_j} P_n}{l_1! \cdots l_j!}.
\eeq
The statement now follows by  applying \eqref{eq:newcumu} for $m\geq 2$ and observing that the $m=1$ disappears because of the first factor at the left-hand side of \eqref{eq:cumull}.
\end{proof}
The latter lemma holds for $z$ sufficiently small. In the next lemma we state an estimate that was of the crucial observations in \cite{BD}, which shows how small $z$ can be.
\begin{lemma} Let $m,n\in \bbN$ and $\mathcal B=({\mathcal B}_{r,s})_{r,s=1}^\infty$ be a bounded operator on $\ell_2(\bbN)$. Then
\beq\label{eq:ineqfromBD}
\left|C_m^{(n)}(\mathcal B) \right| \leq \frac{m^{3/2} {\rm e}^m}{\sqrt{2\pi}}  \|B\|_\infty^{m-2} \|[\mathcal B,P_n]\|_2^2, 
\eeq
where $[B,P_n]=B P_n-P_n B$ is the commutator of $P_n$ and $B$ and $\|\cdot \|_2$ stands for the Hilbert-Schmidt norm.  Hence, there exists a universal constant $A>0$ such that 
\beq \label{eq:oldestimate}
 \exp(-\Tr \mathcal B P_n) \det\left( 1+ ({\rm e}^{z  \mathcal B}-1)P_n\right)\leq \exp  A |z|^2 \|[\mathcal B,P_n]\|_2^2,
\eeq
for $|z | \leq 1/(3\|\mathcal B \|_\infty)$. 

Moreover, 
 For any $m\in \bbN$ and $F$ of finite rank we have 
\beq\label{eq:ineqfromBD2}
\left| D_m(F) \right| \leq \frac{m^{3/2} {\rm e}^m}{\sqrt{2\pi}}  \|F\|_\infty^{m-2} \|[F,P_-]\|_2^2, 
\eeq
\end{lemma}

\begin{proof} The proof was already given in \cite{BD}, but for completeness we will provide the main line of reasoning here. 

The key is that, by only using the identity $P_n^2=P_n$ and the fact that $\Tr AB=\Tr BA$, we have 
\begin{multline*}
\Tr \mathcal B^{l_1} P_n\cdots \mathcal B^{l_j} P_n-\Tr \mathcal\mathcal B^{l_1} P_n\cdots \mathcal B^{l_{j-1}+l_j} P_n
\\
= \Tr P_n \mathcal B^{l_1} P_n\cdots \mathcal B^{l_j} P_n-\Tr P_n \mathcal\mathcal B^{l_1} P_n\cdots \mathcal B^{l_{j-1}+l_j} P_n\\
= \Tr P_n  \mathcal B^{l_1} P_n\cdots \mathcal B^{l_{j-2}} P_n [P_n, \mathcal B^{l_{j-1}}] [P_n,\mathcal B^{l_j}].
\end{multline*}
for $j\geq 2$.  Then we estimate the trace by 
\begin{multline}\label{eq:oldproofa}
\left|\Tr \mathcal B^{l_1} P_n\cdots \mathcal B^{l_j} P_n-\Tr \mathcal B^{l_1} P_n\cdots \mathcal B^{l_{j-1}+l_j} P_n\right|\\
\leq \left\|P_n  \mathcal B^{l_1} P_n\cdots \mathcal B^{l_{j-2}} [P_n, \mathcal B^{l_{j-1}}] [P_n,\mathcal B^{l_j}]P_n\right\|_1\\
\leq  \| \mathcal B\|_\infty^{l_1+\cdots+l_{j-2}}\| [P_n, \mathcal B^{l_{j-1}}]\|_2\| [P_n, \mathcal B^{l_j}]\|_2.
\end{multline}
Here we used that $\|P_n\|_\infty=1$ and the identities for the trace and Hilbert-Schmidt norm as listed in Section \ref{sec:prel}.  By also using $$[\mathcal B^l, P_n]= \sum_{j=1}^{l_j} \mathcal B^{l-j} [\mathcal B,P_n] \mathcal B^{j-1}$$ we obtain 
$$\|[\mathcal B^l, P_n]\|_2^2 \leq l^2\|\mathcal B\|_\infty^{2(l-1)} \|[\mathcal B, P_n]\|_2^2.$$
By substituting this into \eqref{eq:oldproofa} and using $l_j, l_{j-1}\leq m$ we obtain 
$$
\left|\Tr \mathcal B^{l_1} P_n\cdots \mathcal B^{l_j} P_n-\Tr 
\mathcal B^{l_1} P_n\cdots \mathcal B^{l_{j-1}+l_j} P_n\right|
\leq m^2  \| \mathcal B\|_\infty^{m-2}\| [P_n, \mathcal B]\|_2^2.$$
By iterating this $j-1$ times we find 
$$
\left|\Tr \mathcal B^{l_1} P_n\cdots \mathcal B^{l_j} P_n-\Tr 
\mathcal B^{m} P_n\right|
\leq  j m^2  \| \mathcal B\|_\infty^{m-2}\| [P_n, \mathcal B]\|_2^2.$$
By substituting this back into definition of $C^{(n)}_m$ in \eqref{eq:newcumu} we obtain 
$$|C_m^{(n)}(\mathcal B)|\leq
m^2 \| \mathcal B\|_\infty^{m-2}\| [P_n, \mathcal B]\|_2^2 \sum_{j=2}^m  \sum_{\overset{l_1+ l_2+ \cdots l_j=m}{ l_i\geq 1}}\frac{1}{l_1! \cdots l_j!}.  $$
Now \eqref{eq:ineqfromBD} follows from combining the latter with 
 $$\sum_{j=2}^m  \sum_{\overset{l_1+ l_2+ \cdots l_j=m}{ l_i\geq 1}} \frac{1}{l_1! \cdots l_j!} <\frac{m^m}{m!} \leq \frac{{\rm e}^m}{\sqrt{ 2 \pi m }}.$$
 From \eqref{eq:ineqfromBD} we obtain \eqref{eq:oldestimate} with $A=\frac{{\rm e}^2}{\sqrt{2\pi}} \sum_{m=2}^{\infty}m^{3/2}({\rm e}/3)^{m-2}<\infty$.

 The proof of \eqref{eq:ineqfromBD2} is  completely analogous to the proof of \eqref{eq:ineqfromBD} after replacing $\mathcal B $ by $F$ and $P_n$ by $P_-$. 
\end{proof}

\begin{theorem}\label{th:maintheoremhelp}
Let $\mathcal B=({\mathcal B}_{r,s})_{r,s=1}^\infty$ be a banded  matrix defining a bounded operator on $\ell_2(\bbN)$. Assume that $\mathcal B$ has a right limit $\mathcal  B^R$ along a subsequence $\{n_j\}_j$. Then 
\begin{multline}\label{eq:maintheoremhelp}
\lim_{j\to \infty} \exp(-z \Tr \mathcal B P_{n_j}) \det \left(1+({\rm e}^{ z \mathcal B}-1) P_{n_j}\right) \\
= \lim_{M\to \infty} \exp(-z \Tr \mathcal B_M^R P_-) \det \left(1+({\rm e}^{ z \mathcal B^R_M}-1) P_-\right) ,
\end{multline} 
uniformly for $z$ in a sufficiently small neighborhood of the origin. In particular, both limits exists and are analytic in a neighborhood of the origin.
\end{theorem}
\begin{proof}
For $M\in \bbN$  we write
\begin{align}
&\exp \left[\sum_{m=2}^\infty z^m C_m^{(n)}(\mathcal B)\right] = \exp \left[\sum_{m=2}^\infty z^m D_m(\mathcal B_{bM}^R)\right] \exp \left[-\sum_{m=M+1}^\infty z^m D_m(\mathcal B_{bM} ^R)\right] \\ &\quad \times  \exp \left[\sum_{m=2}^{ M} z^m ( C_m^{(n)}(\mathcal B)-D_m(\mathcal B_{bM}^R)\right]\exp \left[\sum_{m= M+1}^\infty z^m C_m^{(n)}(\mathcal B)\right]
\end{align}
Note that by \eqref{eq:ineqfromBD} we have
\begin{multline}\label{eq:tailbound1}
 \sum_{m=M+1}^\infty  |z|^m |C_m^{(n)}(\mathcal B)|
\leq \frac{1}{\sqrt{2 \pi}} \|[\mathcal B,P_n]\|_2^2 \sum_{m=M+1}^\infty m^{3/2} ({\rm e}|z|)^m \|\mathcal B\|_\infty^{m-2},    \end{multline}
and by \eqref{eq:ineqfromBD2} we have
 \begin{multline}\label{eq:tailbound2}
 \sum_{m=M+1}^\infty  |z|^m | D_m(\mathcal B^R_{bM})|
\leq \frac{1}{\sqrt{2 \pi}} \|[\mathcal B_{b M}^R,P_{-}]\|_2^2 \sum_{m=M+1}^\infty m^{3/2} ({\rm e}|z |)^m \|\mathcal B^R_{bM}\|_\infty^{m-2}.    \end{multline}
Since $\mathcal B$ is banded, the commutators in \eqref{eq:tailbound1} and \eqref{eq:tailbound2} are of finite rank and hence it is not difficult to see that the Hilbert-Schmidt norms are bounded in $M$ and $n$. Therefore we have that the  bounds in both \eqref{eq:tailbound1} and \eqref{eq:tailbound2}  tend to zero as $M\to \infty$, uniformly for $z$ in a sufficiently small neighborhood of the origin.

By \eqref{eq:oldestimate} and Montel's Theorem, the sequence  $\{F_{j}(z)\}_{j}$ defined by
$$
F_{j}(z)=\exp \left[\sum_{m=1}^\infty z^m C_m^{(n_j)}(\mathcal B)\right] $$
defines a normal sequence of analytic functions on a small disk around the origin. Hence the exists a subsequence $\{j'\}$ such that $\{F_{j'}\}$ converges to an analytic function $F(z)$ on that disk. But then by   Lemma \ref{lem:fromctoD} we have
\begin{align}\no
&F(z) = \exp \left[\sum_{m=2}^\infty z^m D_m(\mathcal B_{bM}^R)\right] \exp \left[-\sum_{m=M+1}^\infty z^m D_m(\mathcal B_{bM}^R)\right] \\ &\qquad \times \exp \left[\lim_{j'\to \infty} \sum_{m=M+1}^\infty z^m C_m^{(n_{j'})}(\mathcal B)\right]\no
\end{align} 
By  using \eqref{eq:tailbound1} and \eqref{eq:tailbound2}  and taking the limit $M\to \infty$ we obtain
\beq\no
F(z) = \lim_{M\to \infty}  \exp \left[ \sum_{m=2}^\infty z^m D_m(\mathcal B_{b M}^R)\right]
\eeq
In particular, the limit at the right-hand side exists.
By rewriting the right-hand side as in Lemma \ref{lem:rewritingoffredholmdet} using \eqref{eq:defDmF} we get
\beq\no
F(z)=
\lim_{M\to \infty}   {\rm e}^{-z \Tr P_{-} \mathcal B^R_{b M} P_{-} }  \det  \left(I+P_{-} ({\rm e}^{z \mathcal B^R_{b M} } -I)P_{-}\right).
\eeq
Now note that the right-hand side is independent of the subsequence $j'$ and therefore every convergent subsequence has the same limit. Hence the sequence $F_{j}$ converges to $F$.  This proves the statement.
 \end{proof}

\section{Proof of Theorems  \ref{th:fluctuationsbio}, \ref{th:approx} and \ref{th:maintheorem}} \label{sec:proof}

In this section we will prove Theorems \ref{th:fluctuationsbio}, \ref{th:approx} and \ref{th:maintheorem}.  We will start with the first one. 

\begin{proof}[Proof of Theorem \ref{th:fluctuationsbio}] 
The starting point of the proof  is the fact that the biorthogonal ensemble of size forms a determinantal point process with  kernel 
$$K_n(x,y)=\sum_{j=0}^{n-1} \psi_j(x)\phi_j(y).$$
As discussed in Section \ref{sec:cumul} the cumulants for a linear statistic associated to a polynomial $p$ can be written as
$$
\mathcal C_m^{(n)} (p)=\sum_{j=1}^m \frac{(-1)^j}{j} \sum_{\overset{l_1+ l_2+ \cdots l_j=m}{ l_i\geq 1}} \frac{\Tr p^{l_1} K_n\cdots  p^{l_j} K_n}{l_1! \cdots l_j!}.
$$
By exploiting the biorthogonality we  see, after some standard algebraic computations, that
\begin{multline}\label{eq:cumulantexpressionbio}
\mathcal C_m^{(n)} (p)=C^{(n)}_m(p(J^{(n)}))\\=\sum_{j=1}^m \frac{(-1)^j}{j} \sum_{\overset{l_1+ l_2+ \cdots l_j=m}{ l_i\geq 1}} \frac{\Tr \mathcal (p(J^{(n)})) ^{l_1} P_n\cdots \mathcal  (p(J^{(n)}))^{l_j} P_n}{l_1! \cdots l_j!},
\end{multline}
for $m,n \in \bbN$. 

To prove the  statement we take $p= P \circ Q$. Lemma \ref{lem:equivalent} implies now that when computing the limit $j \rightarrow \infty$, we may replace  $P(Q(J^{(n)}))$ by $T(P(s_Q))$, since both have $L(P(s_Q))$ as a right limit.  By applying Corollary~\ref{cor:toeplitzcumu} to $T(P(s_Q))$,  we see that the cumulants, and hence the moments, converge to those of  a normally distributed random variable with the variance as given in the statement.  This proves the theorem. \end{proof}

\begin{proof}[Proof of Theorem \ref{th:maintheorem}]
We start by recalling from \eqref{sec:cumul} that we have 
\[\EE\left[ \exp z X^{(n)}_f\right] =\det \left(1+ ({\rm e}^{z f}-1) K_n \right),\]
where $K_n$ is as in \eqref{eq:defKn}.
Moreover, since $\EE X^{(n)}_f= \Tr f K_n$ we have 
\[\EE\left[ \exp z (X^{(n)}_f- \EE X^{(n)}_f)\right] =\det \left(1+ ({\rm e}^{z f}-1) K_n \right)\exp(-z \Tr f K_n).\]
By \eqref{eq:cumulantexpressionbio} and the expansion of a Fredholm determinant in terms of the cumulants as explained in Section \ref{sec:cumul},  we find  
$$\EE\left[ \exp z (X^{(n)}_f- \EE X^{(n)}_f)\right] =\det \left(1+ ({\rm e}^{z f(J)}-1) P_n \right)\exp(-z \Tr f(J) P_n).$$
Since $J$ is banded and $f$ is a polynomial, also $f(J)$ is banded. Moreover, it is not hard to check that $f(J^R)$ is a right-limit of $f(J)$.  Hence Theorem \ref{th:maintheorem} follows directly from Theorem \ref{th:maintheoremhelp}.
\end{proof}

\begin{proof}[Proof of Theorem \ref{th:approx}]
For simplicity we will assume that the the subsequence is $\bbN$. Also, by possibly increasing the set $E$ we can always assume it is an interval and that $s(\mathbb T) \subset E$ where $\mathbb T= \{|z|=1\}$. Finally, by replacing $f$ and $p$ below by $f \circ Q$ and $p \circ Q$, we see that it is enough to prove the theorem for $Q=id$, which we do for notational simplicity.

The first step of the proof is showing that the limit of the variance is continuous. Note that $\hat{f}_k$ with $k$ as defined in \eqref{eq:cor:coef} are the Fourier coefficients of $f(s(w))$. Moreover, by integration by parts we see that $k \hat{f}_k$ are the Fourier coefficients of $f'(s(w))s'(w)$. 
\begin{multline*}
\sum_{k=1}^\infty k \hat{f}_k \hat{f}_{-k} -\sum_{k=1}^\infty k \hat{g}_k \hat{g}_{-k} = \sum_{k=1}^\infty k (\hat{f}_k-\hat{g}_k)  \hat{f}_{-k} +\sum_{k=1}^\infty \hat{g}_k  k(\hat{f}_{-k}-\hat{g}_{-k}) 
\end{multline*}
By using the  Cauchy-Schwarz inequality and the unitarity of the Fourier transform we thus see
\begin{multline*}
\left|\sum_{k=1}^\infty k \hat{f}_k \hat{f}_{-k} -\sum_{k=1}^\infty k \hat{g}_k \hat{g}_{-k} \right|
\leq\left( \|f(s(w))\|_2 + \|g(s(w))\|_2\right) \left\|\left(f'(s(w))-g'(s(w))\right) s'(w)\right\|_2,
\end{multline*}
where $\|\cdot\|_2$ stands for the $\mathbb L_2$ norm on the unit circle. 
Hence there exists a constant $c_1>0$ such that 
\beq\label{eq:approx:cont1}
\left|\sum_{k=1}^\infty k \hat{f}_k \hat{f}_{-k} -\sum_{k=1}^\infty k \hat{g}_k \hat{g}_{-k} \right|
\leq c_1\left(  \sup_{x \in E} |f(x)| +  \sup_{x \in E} |g(x)|\right)  \sup_{x \in E} |f'(x)-g'(x)|.
\eeq
This is the continuity result that we need for the limiting expression of the variance.

For finite $n$ we have a similar estimate: Let  $f\in C^1$ be such that there exists a constant $C$ such that
\beq \label{eq:boundonf}
|f(x)| \leq  C (1+|x|)^k, 
\eeq
for $x\in \bbR$. Then, by the assumptions of the theorem, we have
 \begin{multline}\label{eq:boundonintegralf}
\left|\int_\bbR f(x)^2 K_n(x,x) {\rm d} x \right|
\leq  \sup_{t \in E} |f(t)|^2 \int_E   K_n(x,x) {\rm d} x +o(1/n)
=\mathcal O(n),
\end{multline} 
as $n \to \infty$. 

Writing the variance as
\begin{multline*}
\Var X_f^{(n)}= 
 \frac12 \iint_{\bbR\times \bbR} ({f(x)-f(y)})^2 K_n(x,y) K_n(y,x) {\rm d}\mu(x) {\rm d}\mu(y) \\
=\frac12 \iint_{E\times E}({f(x)-f(y)})^2 K_n(x,y) K_n(y,x) {\rm d}\mu(x) {\rm d}\mu(y)
\\
+\frac12 \iint_{ (\bbR\setminus E)\times E}({f(x)-f(y)})^2 K_n(x,y) K_n(y,x) {\rm d}\mu(x) {\rm d}\mu(y)\\+\frac12 \iint_{\bbR \times (\bbR\setminus E)}({f(x)-f(y)})^2 K_n(x,y) K_n(y,x) {\rm d}\mu(x) {\rm d}\mu(y),
\end{multline*}
we see that
\begin{multline*}
0 \leq \Var X_f^{(n)}- \frac12 \iint_{E\times E}({f(x)-f(y)})^2 K_n(x,y) K_n(y,x) {\rm d}\mu(x) {\rm d}\mu(y)
\\
\leq \iint_{\bbR \times (\bbR\setminus E)}({f(x)-f(y)})^2 K_n(x,y) K_n(y,x) {\rm d}\mu(x) {\rm d}\mu(y).
\end{multline*}

Now, since we have $K_n(x,y) K_n(y,x) \leq K_n(x,x) K_n(y,y)$ we also have
\begin{multline} \label{eq:approx:variance1}
0 \leq \Var X_f^{(n)}- \frac12 \iint_{E\times E}({f(x)-f(y)})^2 K_n(x,y) K_n(y,x) {\rm d}\mu(x) {\rm d}\mu(y)
\\
\leq \iint_{\bbR \times (\bbR\setminus E)}({f(x)-f(y)})^2 K_n(x,x) K_n(y,y) {\rm d}\mu(x) {\rm d}\mu(y)\\
\leq 2 \iint_{\bbR \times (\bbR\setminus E)} \left(f(x)^2+f(y)^2\right) K_n(x,x) K_n(y,y) {\rm d}\mu(x) {\rm d}\mu(y)=o(1),
\end{multline}
as $n\to \infty$, where in the last step we used \eqref{eq:boundonf}, \eqref{eq:boundonintegralf} and the condition of the theorem.
Moreover,
\begin{multline}\label{eq:approx:variance2}
 \iint_{E\times E} ({f(x)-f(y)})^2 K_n(x,y) K_n(y,x) {\rm d}\mu(x) {\rm d}\mu(y) \\
=\iint \left(\frac{f(x)-f(y)}{x-y}\right)^2 (x-y)^2K_n(x,y)K_n(y,x) {\rm d}\mu(x) {\rm d}\mu(y) \\
\leq  \left(\sup_{x\in E} |f'(x)|\right)^2\iint  (x-y)^2K_n(x,y) K_n(y,x) {\rm d}\mu(x) {\rm d}\mu(y)\\=2\left(\sup_{x\in E} |f'(x)|\right)^2\Var X_{e_1}^{(n)}
\end{multline}
where $e_1(x)=x$ is the linear monomial. 

By combining \eqref{eq:approx:variance1} and \eqref{eq:approx:variance2} we thus see that there exists a constant $c_2>0$ such that
\beq\label{eq:approx:cont2}
\limsup_{n\to \infty}  \Var X_f^{(n)} \leq c_2 \left(\sup_{x\in E} |f'(x)|\right)^2.
\eeq

We are now ready for the final argument of the proof. 
Let $\eps>0$.  Then, by compactness of $E$, we can find a polynomial $p$ such that 
\beq \label{eq:epsilonbound}
\sup_{x\in E} |f'(x)-p'(x)| \leq \eps
\eeq
and also 
\beq \label{eq:boundforp}
\sup_{x\in E} |f(x)-p(x)| \leq \eps \cdot |E|.
\eeq
For such $p$ write
\begin{multline} \label{eq:approxbigplan}
\left| \EE \left[\exp {\rm i} t ( X_f^{(n)}- \EE X_f^{(n)})  \right]- \exp\left(-\frac12t^2 \sum_{k=1}^\infty k f_k f_{-k} \right)\right|\\
\leq \left| \EE \left[\exp {\rm i} t ( X_f^{(n)}- \EE X_f^{(n)})  \right]- \EE \left[\exp {\rm i} t ( X_p^{(n)}- \EE X_p^{(n)})  \right]\right|\\
+\left| \EE \left[\exp {\rm i} t ( X_p^{(n)}- \EE X_p^{(n)})  \right]- \exp\left(-\frac12 t^2 \sum_{k=1}^\infty k p_k p_{-k} \right)\right|\\
+\left|\exp\left(-\frac12 t^2 \sum_{k=1}^\infty k p_k p_{-k} \right)\ - \exp\left(-\frac12t^2 \sum_{k=1}^\infty k f_k f_{-k} \right)\right|.
\end{multline}

We also recall the following standard argument for real valued random variables $X$ and $Y$, and $t\in \bbR$:
\begin{multline*}
\left|\EE [\exp {\rm i} t X]-\EE[\exp {\rm i} t Y]\right|\leq  \left|\EE [\exp {\rm i} t X-\exp {\rm i} t Y]\right|\\ \leq   \left|\EE [(\exp {\rm i} t (X-Y))-1 \exp {\rm i} t Y]\right|\leq  \EE \left[ \left|\exp {\rm i} t (Y-X))-1\right|\right]\\\leq  |t| \EE \left[ \left| X-Y|\right|\right]\leq  |t|\left( \EE \left[ \left(X-Y\right)^2\right]\right)^{1/2}.
\end{multline*}
Moreover, if $\EE X= \EE Y$ then  $\EE \left[ \left(X-Y\right)^2\right]= \Var (X-Y)$.  Hence, by \eqref{eq:approx:cont2} and \eqref{eq:epsilonbound} we have
 \begin{multline}\label{eq:approx:est1}
 \limsup_{n\to \infty} \left|\EE\left[ \exp {\rm i } t  (X^{(n)}_f- \EE X^{(n)}_f)\right]-\EE\left[ \exp {\rm i } t  (X^{(n)}_{p}- \EE X^{(n)}_{p})\right] \right| \\
 \leq c_2  \sup_{x\in E} |f'(x)-p'(x)\|_\infty \leq \eps c_2, 
 \end{multline}
 for $n\in \bbN$ and $t \in \bbR$.  Moreover, by Corollary \ref{cor:fluctuationsbio} and the fact that $p$ is a polynomial,
 \beq\label{eq:approx:est2}
 \lim_{n\to \infty} \EE \left[\exp {\rm i} t ( X_p^{(n)}- \EE X_p^{(n)})  \right]= \exp\left(-\frac12 t^2 \sum_{k=1}^\infty  k \hat{p}_k \hat{p}_{-k} \right).
 \eeq
 Finally, by \eqref{eq:approx:cont1} and \eqref{eq:epsilonbound} there exists a constant $c_3>0$ such that 
 \beq\label{eq:approx:est3}
\left|\exp\left(-\frac12 t^2 \sum_{k=1}^\infty k \hat{p}_k \hat{p}_{-k} \right)\ - \exp\left(-\frac12t^2 \sum_{k=1}^\infty k \hat{f}_k \hat{f}_{-k} \right)\right|\leq c_3( \|f(s(w))\|_2+ \| p(s(w)\|) \eps.
\eeq 
By substituting \eqref{eq:approx:est1}, \eqref{eq:approx:est2} and \eqref{eq:approx:est3} into \eqref{eq:approxbigplan} we obtain 
\begin{multline*} 
\limsup_{n\to \infty} \left| \EE \left[\exp {\rm i} t ( X_f^{(n)}- \EE X_f^{(n)})  \right]- \exp\left(-\frac12t^2 \sum_{k=1}^\infty k \hat{f}_k \hat{f}_{-k} \right)\right|\\
\leq c_2 \eps +c_3( \sup_{x \in E} |f(x)|+  \sup_{x \in E} |p(x)|) \eps \leq c_2 \eps + c_3  \sup_{x \in E} |f(x)|  (2+\eps |E|) \eps,
\end{multline*}
where we also used \eqref{eq:boundforp} in the last step. Since the constants do not depend on $p$ we can take the limit $\eps \downarrow 0$ and this proves the theorem.
\end{proof}

\section{Examples}

In this section we demonstrate the power of our approach by studying the application of our results to some interesting examples. 

The first subsection presents examples of orthogonal polynomial ensembles with compactly supported measures that are purely singular with respect to Lebesgue measure, but still satisfy a Central Limit Theorem. We also present examples of orthogonal polynomial ensembles that have different Central Limit Theorems along different subsequences. 

In subsection \ref{sec:Hahn} we discuss the application of our results to Hahn polynomial ensembles and the implications for lozenge tilings of the hexagon. 

In subsection \ref{sec:two} we discuss the application of our results to the two matrix model.
 
\subsection{Some Examples of Orthogonal Polynomial Ensembles} \label{sec:OPEexamples}

The focus on recursion coefficients is extremely natural from the point of view of the spectral theory of Jacobi matrices, since the orthogonality measure for a sequence of orthogonal polynomials is also the spectral measure for the associated Jacobi matrix (see, e.g., \cite{deift}). From this point of view, the simplest model is $a_n \equiv 1$ and $b_n \equiv 0$. These are the recurrence coefficients for the (scaled) Chebyshev polynomials of the second kind, for which the orthogonality measure is $\textrm{d} \mu(x)=\frac{\sqrt{4-x^2}}{2\pi}\chi_{[-2,2]}(x) \textrm{d} x$. This Jacobi matrix can also be thought of as the discrete Laplacian on the half line.

Thus a Jacobi matrix with $a_n \rightarrow 1$ and $b_n \rightarrow 0$ is a compact perturbation of the Laplacian and so, by Weyl's Theorem \cite[Theorem XIII.14]{rs4}, the spectral measure $\mu$, for such a matrix, has $\sigma_{\textrm{ess}}(\mu)=[-2,2]$. Examples of such matrices with purely singular (with respect to Lebesgue measure) spectral measures have been extensively studied in the Schr\"odinger operator community. For a review of some such examples see \cite{last-review}.

\begin{example} \label{ex:sparse1}
One example, whose analog in the continuum case was first considered by Pearson \cite{pearson}, is that of a decaying `sparse' perturbation of the Laplacian. In this case 
\beq \no
a_n \equiv 1,
\eeq
and 
\beq \no
b_n =\left \{ \begin{array}{cc} \widetilde{b}_j & n=N_j \\
0 & \textrm{otherwise}, \end{array} \right.
\eeq 
where $\widetilde{b}_j \neq 0$, $\widetilde{b}_j \rightarrow 0$ and $N_j$ satisfies $\frac{N_j}{N_{j+1}} \rightarrow 0$ as $j \rightarrow \infty$. In this case, \cite[Theorem 1.7]{kls} says that if $\sum_{j=1}^\infty \widetilde{b}_j^2 < \infty$ then $\mu$ is purely absolutely continuous on $(-2,2)$ and if $\sum_{j=1}^\infty \widetilde{b}_j^2=\infty$ then $\mu$ is purely singular continuous on $(-2,2)$. Thus, by Theorem \ref{th:theorem1} we see there are many singular measures such that the associated orthogonal polynomial ensemble satisfies a Central Limit Theorem. 

It is also worth mentioning here that sparse perturbations of the Laplacian were used in \cite{Bsine} to construct singular continuous meausres such that the associated Christoffel-Darboux kernel satisfies sine kernel asymptotics. For a review of some of the work on the spectral theory of Schr\"odinger operators with decaying sparse potentials see \cite{last-review}.
\end{example}

\begin{example} \label{sparse2}
Fix $\widetilde{b} \neq 0$ and let 
\beq \no
a_n \equiv 1,
\eeq
and 
\beq \no
b_n =\left \{ \begin{array}{cc} \widetilde{b} & n=N_j \\
0 & \textrm{otherwise}, \end{array} \right.
\eeq 
where $N_j$ is any sequence satisfying $N_{j+1}-N_j \rightarrow \infty$ as $j \rightarrow \infty$. By \cite[Theorem 1.7]{S304} the spectral measure $\mu$ of the associated Jacobi matrix has $\mu_{\textrm{ess}}=[-2,2]$ together with perhaps another point outside of $[-2,2]$. In addition, by \cite[Corollary 1.5]{Remling}, $\mu$ is purely singular. 

Note that for any sequence $\{ n_j\}_{j=1}^\infty$, satisfying $\lim_{j \rightarrow \infty} \inf_{\ell}|n_j-N_\ell|=\infty$, the right limit along $n_j$ is a Laurent matrix . Thus, by Theorem \ref{th:theorem1} along any such sequence the associated orthogonal polynomial ensemble satisfies a Central Limit Theorem. 

The right limit along $n_j \equiv N_j$, however, is a rank one perturbation of a Laurent matrix given by 
\beq \no
a_n \equiv 1,
\eeq
and 
\beq \no
b_n =\left \{ \begin{array}{cc} \widetilde{b} & n=0 \\
0 & \textrm{otherwise}. \end{array} \right.
\eeq 
Theorem \ref{th:maintheorem} gives a formula for the limt of fluctuations in this case. It would be interesting to give a simpler formula than \eqref{eq:maintheorem} for the limit, at least for this case of a rank one perturbation of a Laurent matrix.
\end{example}

The following example shows that an orthogonal polynomial ensemble with a measure supported on an interval can satisfy several different Central Limit Theorems.   

\begin{example} \label{ex:example1}
Partition $\bbN$ into successive blocks $A_1,C_1,A_2,C_2,\ldots$ whose size is given by 
\beq \no
\#(A_j)=3^{j^2}, \qquad \#(C_j)=2^{j^2}.
\eeq
Now let the sequence $\{a_k,b_k\}_{k=1}^\infty$ be given by  
\beq \no
a_k=\left \{ \begin{array}{cc} 1 & k \in A_j \\
\frac{1}{2} & k \in C_j, \end{array}\right.
\eeq
and $b_k \equiv 0$, and let $\mu$ be the corresponding spectral measure. Then, by the fact that $\|J \| \leq 2$, $\supp(\mu)=\textrm{spectrum}(J) \subseteq [-2,2]$. Moreover, by constructing approximate eigenfunctions, $\supp(\mu)=[-2,2]$. 

Now, consider the corresponding orthogonal polynomial ensemble. If $n_j$ is taken in the middle of the $A_j$ blocks then by Theorem \ref{th:theorem1}, for any real-valued $f\in C^1(\bbR)$
\begin{align}
X_f^{(n_j)} -\EE X_f^{(n_j)} \to N\left(0,\sum_{k \geq 1} k  |\hat f_k|^2 \right), \qquad \text{ as } j \to \infty,
\end{align}
in distribution, where  
\beq \no
 \hat f_k=\frac{1}{2\pi} \int_0^{2\pi}  f(2  \cos  \theta) {\rm e}^{-{\rm i} k \theta} {\rm d}\theta.
\eeq
On the other hand, if $n_j$ is taken in the middle of the $C_j$ blocks then for any real-valued $f\in C^1(\bbR)$
\begin{align}
X_f^{(n_j)} -\EE X_f^{(n_j)} \to N\left(0,\sum_{k \geq 1} k  |\hat {\hat {f_k}}|^2 \right), \qquad \text{ as } j \to \infty,
\end{align}
in distribution, where  
\beq \no
\hat{\hat {f_k}}=\frac{1}{2\pi} \int_0^{2\pi}  f( \cos  \theta) {\rm e}^{-{\rm i} k \theta} {\rm d}\theta.
\eeq

Thus, the fluctuations in this example have at least two different Gaussian limits. Note that by the Stahl-Totik theory of regular measures \cite{StT}, for any continuous $f$
\beq \no
\frac{1}{n}\EE X_f^{(n)} \to \int f(x) {\rm d}\nu(x)
\eeq 
where $\nu$ is the potential theoretic equilibrium measure for the interval $[-2,2]$.
\end{example}

Example \ref{ex:example1} is Example 4.1 from \cite{BLS}, where it was constructed for a different purpose. We note that by \cite{Remling}, the measure $\mu$ is purely singular with respect to Lebesgue measure. Example 5.1 from \cite{BLS}, which we describe next, shows that $\mu$ may be supported on an interval, have a nontrivial absolutely continuous component, and yet the corresponding orthogonal polynomial ensemble may have several different Gaussian limits. In fact, a \emph{continuum} of different Gaussian limits is attained. 

\begin{example}\label{ex:example2} 
Let $b_n \equiv 0$ and $a_n$ be described as follows:
Partition $\{1,2,\dots\}$ into successive blocks $A_1,B_1,C_1,D_1,A_2,B_2, \dots$, where
\begin{equation} \no
\#(A_j)=3^{j^2} \quad \#(C_j) =2^{j^2} \quad \#(B_j)=\#(D_j)=j^6 -1
\end{equation}
On $A_j$, $a_n\equiv 1$, on $C_j$, $a_n\equiv\frac{1}{2}$, and on $B_j$ and
$D_j$, $\log(a_n^2)$ interpolates linearly from $\log(\frac{1}{4})$ to $\log(1)$,
that is, for $n\in B_j$,
\begin{equation} \no
\frac{a_n^2}{a_{n-1}^2} = c_j
\end{equation}
and for $n\in D_j$,
\begin{equation} \no
\frac{a_{n-1}^2}{a_n^2} = c_j
\end{equation}
where
\begin{equation} \no
c_j^{j^6} =\frac{1}{4}.
\end{equation}
Theorem 5.2 in \cite{BLS} says that in this case $\mu$ is purely absolutely continuous on $(-1,1)$  and purely singular on $[-2,2]\setminus(-1,1)$. 

As in Example \ref{ex:example1} above, it is easy to see that the two Toeplitz matrices, one with $a_n \equiv 1$ and one with $a_n \equiv 1/2$ ($b_n \equiv 0$ in any case) are right limits. In fact, in this case, for any $a \in [1/2,1]$, there is a Toeplitz right limit with $a_n \equiv a$. This is because along the blocks $B_n$ and $D_n$, $a_n$ varies at an increasingly slower rate between $1/2$ and $1$. There are no other right limits, since for any sequence, $\{n_j\}_{j=1}^\infty$, along which $a_{n_j}$ approaches a constant, all the $a_n$'s in any finite-size block around $n_j$ approach the same constant. 

It follows that for any $a \in [1/2,1]$ there is a sequence $\{n_j(a)\}_{j=1}^\infty$ such that for any real-valued $f\in C^1(\bbR)$
\begin{align}
X_f^{\left(n_j(a) \right)} -\EE X_f^{\left(n_j(a)\right)} \to N\left(0,\sum_{k \geq 1} k  |\hat{f}_k(a)|^2 \right), \qquad \text{ as } j \to \infty,
\end{align}
in distribution, where  
\beq \label{eq:coefficientsf}
\hat{f}_k(a)=\frac{1}{2\pi} \int_0^{2\pi}  f(2a \cos  \theta) {\rm e}^{-{\rm i} k \theta} {\rm d}\theta.
\eeq  
Here again, as in example \ref{ex:example1}, the measure $\mu$ is regular and so for any continuous $f$
\beq \no
\frac{1}{n}\EE X_f^{(n)} \to \int f(x) {\rm d}\nu(x)
\eeq 
where $\nu$ is the potential theoretic equilibrium measure for the interval $[-2,2]$.
\end{example}

\subsection{The Hahn ensemble and lozenge tilings of a hexagon} \label{sec:Hahn}

In the next example, we discuss a  discrete orthogonal polynomial ensemble. Such  ensembles  appear naturally in random tilings of planar domains. Examples are the classical discrete polynomials  such as Charlier, Hahn, Krawtchouk and Meixner polynomials  \cite{BKMM,Jannals2}.   Since the Jacobi matrices corresponding to  classical discrete orthogonal polynomials (after a rescaling) have a right limit that is a Laurent matrix, we have that Theorem \ref{th:theorem1} applies to these ensembles. As an illustration we work out the case of the Hahn ensemble that describes random lozenge tilings of a hexagon. We will also briefly discuss the connection with the study of the two dimensional fluctuations of such systems and the appearance of the Gaussian Free Field.

For  $\alpha, \beta >-1$ and $N\in \bbN$. The \emph{Hahn weight} $w_N^{(\alpha,\beta)}$ is defined as 
\beq\label{eq:hahn}
w_N^{(\alpha,\beta)}(x)= \begin{pmatrix} \alpha+x\\ x \end{pmatrix} \begin{pmatrix} \beta+N-x\\ N-x \end{pmatrix},
\eeq
on the nodes $ x=0,1,\ldots,N$.  The corresponding  orthogonal polynomials are called the \emph{Hahn polynomials} and are given by 
\beq\no
Q^{(\alpha,\beta)}_{N,n}(x)= {}_3 F_2\left(\begin{array}{c}-n,n+\alpha+\beta+1,-x\\ \alpha+1,-N  \end{array};1\right),
\eeq
for $n=0,1,\ldots,N$. For properties of Hahn polynomials we refer to \cite{koekoek}. The \emph{Hahn ensemble} is defined as the discrete orthogonal polynomial ensemble induced by the weight \eqref{eq:hahn} on the nodes $x=0,1,\ldots,N$.  The Hahn ensemble appears in Lozenge tilings of the hexagonal lattice that we will now briefly describe. For more details we refer to \cite{BKMM,Jannals2}.  

Fix $\gothic{a}, \gothic{b},\gothic{c} \in \bbN$  and consider the hexagon with corners at $(0,0)$, $(\gothic{a},-\gothic{a})$, $(\gothic{a}+\gothic{b},\gothic{b}-\gothic{a})$, $(\gothic{a}+\gothic{b}, \gothic{b}-\gothic{a}+ 2\gothic{c})$, $(\gothic{b}, \gothic{b}+ 2\gothic{c})$ and $(0,  2\gothic{c})$. Next we introduce lozenges of three different types. Set $e=(1,0)$ and $f=(0,1)$. Then a type I lozenge is spanned by $e+f $ and $2f$, a type II lozenge by $e-f$ and $-2f$, and a type III lozenge by $e-f$ and $e+f$.  Note that  the endpoints of the lozenges are always on the square grid. We will consider tilings of the $\gothic a \gothic b \gothic c$-hexagon by the three type of lozenges.  By symmetry we assume without loss of generality that $\gothic b \geq \gothic a$. 

Let us now consider  a random tiling by taking the uniform measure on all possible tilings of the hexagon. Our particular interest is in  the point  process defined by the type III lozenges whose centers lie at a vertical section $\{m\} \times \bbN$ for a given $m\in \{1,\ldots, \gothic a + \gothic b -1\}$. To describe this process, let $L_m$ be the number of type III lozenges on such a section. After some considerations, we find that
\beq
L_m= \begin{cases}
m, &  \text{ if }  0  \leq m\leq \gothic a,\\
\gothic a & \text{ if }  \gothic a  <m \leq \gothic b,\\
\gothic a+\gothic b-m& \text{ if } \gothic  b < m \leq \gothic a+\gothic b,\\
\end{cases}\eeq
Note that $\gothic c +L_m+1$  is the total number of points on the square lattice on the vertical line at $m$  that lie inside the hexagon and define 
$$\gothic a_m=|\gothic a-m|, \quad \text{and} \quad \gothic b_m=|\gothic b-m|.$$
Then it was proved by Johansson \cite{Jannals2} (see also \cite{BKMM}) that the type III lozenges with centers at the line $m \times \bbN$ form a Hahn ensemble of size $L_m$ with parameters $N=\gothic c+L_m$ and $(\alpha,\beta)=(\gothic a_m+1,\gothic b_m+1)$. 

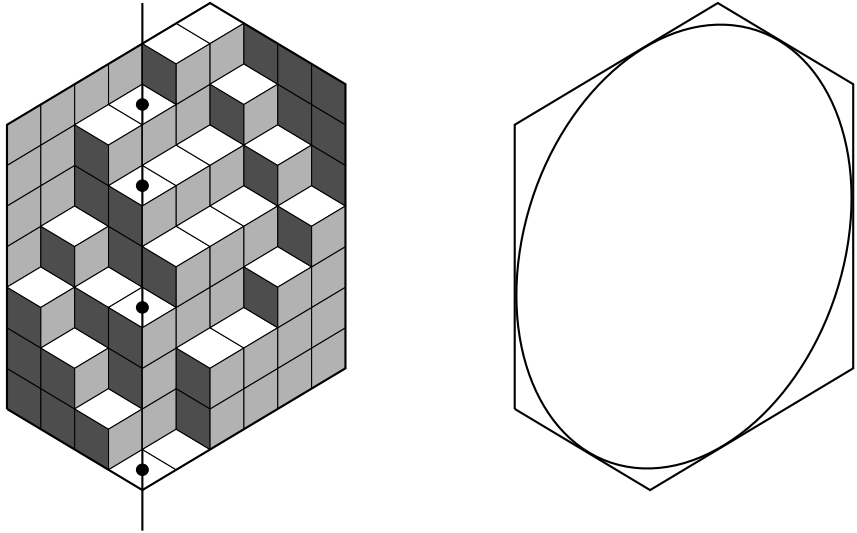
\begin{figure}
\begin{center}
\begin{tikzpicture}[xscale=0.45,yscale=0.27]

\draw (0,0) \lozd;
\draw (0,2) \lozd;
\draw (0,4) \lozd;
\draw (0,6) \lozu;
\draw (0,8) \lozu;
\draw (0,10) \lozu;
\draw (0,12) \lozu;

\draw (0,6) \lozr;

\draw (1,-1) \lozd;
\draw (1,1) \lozd;
\draw (1,3) \lozu;
\draw (1,7) \lozd;
\draw (1,9) \lozu;
\draw (1,11) \lozu;
\draw (1,13) \lozu;

\draw (1,3) \lozr;
\draw (1,9) \lozr;

\draw (2,-2) \lozd;
\draw (2,0) \lozu;
\draw (2,4) \lozd;
\draw (2,6) \lozu;
\draw (2,10) \lozd;
\draw (2,12) \lozd;
\draw (2,14) \lozu;

\draw (2,0) \lozr;
\draw (2,6) \lozr;

\draw (2,14) \lozr;

\draw (3,-3) \lozu;
\draw (3,1) \lozd;
\draw (3,3) \lozd;

\draw (3,7) \lozd;
\draw (3,9) \lozd;
\draw (3,11) \lozu;
\draw (3,15) \lozu;

\draw (3,-3) \lozr;
\draw (3,5) \lozr;
\draw (3,11) \lozr;
\draw (3,15) \lozr;

\filldraw (4,-3) ellipse (5pt and 8pt);
\filldraw (4,5)  ellipse (5pt and 8pt);
\filldraw (4,11) ellipse (5pt and 8pt);
\filldraw (4,15)  ellipse (5pt and 8pt);

\draw[thick] (4,-6) -- (4,20);
\draw (4,-2) \lozu;
\draw (4,0) \lozu;
\draw (4,2) \lozu;
\draw (4,6) \lozd;
\draw (4,8) \lozu;
\draw (4,12) \lozu;
\draw (4,16) \lozd;

\draw (4,-2) \lozr;
\draw (4,8) \lozr;
\draw (4,12) \lozr;
\draw (4,18) \lozr;

\draw (5,-1) \lozd;
\draw (5,1) \lozd;
\draw (5,3) \lozu;
\draw (5,5) \lozu;
\draw (5,9) \lozu;
\draw (5,13) \lozu;
\draw (5,15) \lozu;

\draw (5,3) \lozr;
\draw (5,9) \lozr;
\draw (5,13) \lozr;
\draw (5,19) \lozr;

\draw (6,-2) \lozu;
\draw (6,0) \lozu;
\draw (6,4) \lozu;
\draw (6,6) \lozu;
\draw (6,10) \lozu;
\draw (6,14) \lozd;
\draw (6,16) \lozu;

\draw (6,4) \lozr;
\draw (6,10) \lozr;
\draw (6,16) \lozr;

\draw (7,-1) \lozu;
\draw (7,1) \lozu;
\draw (7,5) \lozd;
\draw (7,7) \lozu;
\draw (7,11) \lozd;
\draw (7,13) \lozu;
\draw (7,17) \lozd;

\draw (7,7) \lozr;
\draw (7,13) \lozr;

\draw (8,0) \lozu;
\draw (8,2) \lozu;
\draw (8,4) \lozu;
\draw (8,8) \lozd;
\draw (8,10) \lozu;
\draw (8,14) \lozd;
\draw (8,16) \lozd;

\draw (8,10) \lozr;;

\draw (9,1) \lozu;
\draw (9,3) \lozu;
\draw (9,5) \lozu;
\draw (9,7) \lozu;
\draw (9,11) \lozd;
\draw (9,13) \lozd;
\draw (9,15) \lozd;

\draw[rotate around={83:(20,8)},thick] (20,8)  ellipse(11cm and 4.8cm);
\draw[ thick] (0,0)--(0,14)--(6,20)--(10,16)--(10,2)--(4,-4)--(0,0);
\draw[ thick] (15,0)--(15,14)--(21,20)--(25,16)--(25,2)--(19,-4)--(15,0);
\end{tikzpicture}
\end{center}
\caption{The left figure shows a lozenge tiling of a hexagon and  the centers of the type III lozenges along a vertical section. The right figure shows the separation of the hexagon into the disordered region and the frozen corners.}
\label{fig:tiling}
\end{figure}

We are particularly interested in the asymptotic properties of this process  when the hexagon is large. That is,  we take $\gothic c\to \infty$ such that $\gothic a/ \gothic c$ and $\gothic b/ \gothic c$ converge to positive real numbers. In this limit, the hexagon will be separated into a number of regions. There is an ellipse that touches each of the faces of the hexagon, (see the right hexagon in Figure \ref{fig:tiling}). Inside the ellipse the system is disordered (also referred to as the liquid region). In the remaining  corner pieces of the hexagon the tiling is frozen, in the sense that each piece has only one type of lozenges. To analyze the fluctuations in the  disordered region we intersect the hexagon with a vertical line with horizontal coordinate $m$ such that $m/\gothic c $ tends to some positive value and consider the lozenges of type III on that line. Due to the fact that these lozenges form a Hahn ensemble we have that the following result, which is a straightforward  corollary to Theorem \ref{th:theorem1} and basic properties of the Hahn polynomials, provides us with a CLT for that point process.

\begin{theorem} \label{th:hahn}
Let $A,B\geq 0$ and $t>1$. Consider the (rescaled) Hahn ensemble of size $n$. As $n\to \infty$, take $\alpha,\beta,N$ such that  $$\begin{cases}
\frac{\alpha}{n} \to A ,\\
\frac{\beta}{n}\to B,\\
\frac{N}{n}\to t.
\end{cases}$$
Then for any $f\in C^1([0,1])$, the scaled linear statistic $$X_f^{(n)} = \sum_{j=1}^n f(x_j/N),$$
with $x_j$  randomly chosen from the Hahn ensemble of size $n$,  satisfies  
$$
X_f^{(n)} -\EE X_f^{(n)} \to N\left(0,\sum_{k \geq 1} k  |\hat f_k|^2 \right), \qquad \text{ as } n \to \infty,
$$
in distribution, where   the coefficients $\hat f_k$ are defined as 
$$
 \hat f_k=\frac{1}{2\pi} \int_0^{2\pi}  f\left(2 a\cos  \theta+b\right) {\rm e}^{-{\rm i} k \theta} {\rm d}\theta,
$$
for $k\geq 1$, and 
\beq
\label{eq:limithahn}
\begin{split}
a&=\frac{\sqrt{(t-1)(1+A+B)(1+A) (1+A+B+t)(1+B)}}{t (2+A+B)^2} \\
b&=\frac{(t-1)(1+A+B)(1+A)+(1+A+B+t)(1+B)}{t(2+A+B)^2}.
\end{split}
\eeq
\end{theorem}

\begin{proof}
Rescale the nodes so that they are on $[0,1]$ for all $N\in \bbN$ and consider the following measure
\beq
 \tilde \mu(x)= \sum_{x=0}^{N} w_N^{(\alpha,\beta)}(x) \delta_{x/N}.
\eeq
The orthonormal polynomials $p_n$ with respect to this measure are  the rescaled and normalized Hahn polynomials. They satisfy the recurrence \eqref{eq:oprecur} with 
\begin{align}\no
\begin{split}
b_n&=\frac{(N-n)(n+\alpha+\beta+1)(n+\alpha+1)}{N(2n+\alpha+\beta+1)(2n+\alpha+\beta+2)}\\
&\quad \quad +\frac{n(n+\alpha+\beta+N+1)(n+\beta)}{N (2n+\alpha+\beta)(2n+\alpha+\beta+1)}\\
a_n&=\frac{n(n+\alpha+\beta+N+1)(n+\beta)}{N (2n+\alpha+\beta)(2n+\alpha+\beta+1)} \\
&\quad  \quad \times \sqrt{\frac{(N-n)(n+\alpha+\beta)(\alpha+n)(2n+\alpha+\beta+1)}{n(n+\alpha+\beta+N+1)(\beta+n)(2n+\alpha+\beta-1)}},
\end{split}
\end{align}
which can be derived  from  the properties in \cite{koekoek}. One readily verifies that $a_n \rightarrow a$ and $b_n \rightarrow b$ as $n \rightarrow \infty$, with $a$ and $b$ as in \eqref{eq:limithahn} and, therefore, the statement is a corollary to Theorem \ref{th:theorem1}.
\end{proof}

Tilings of planar domains are often effectively described by the so-called height function, which  associates a  surface to the tiling.  In the lozenge tiling of the hexagon we define the height function by 
$$h_m(x)= \text{Number of type III lozenges below } (m,x). $$
Note that this definition of may differ from the definition that is used in other places in the literature by a adding a deterministic function. This, however, does not have an effect on the fluctuations of the height function.  As the size of the hexagon grows large, the height function converges to a deterministic function, which is the integral of the density of type III lozenges (note that this density can be computed from the recurrence coefficients using the procedure in \cite{KV}).  Theorem \ref{th:hahn} describes the fluctuations of the height function in the following way. Let $f$ have compact support in $[0,1]$, then we have, with $\chi_A$ the characteristic function for the set $A$, 
\begin{multline}
\sum_{x=0}^N h_m(x) (f(x/N)-f((x-1)/N) = \sum_{x=0}^N  f(x/N) (h_m(x+1)-h_m(x)) \\=  \sum_{x=0}^N  f(x/N) \chi_{ \{\text{centers of type III lozenges}\}} (m,x)= X_f^{(n)}.
\end{multline}
Thus we see that Theorem \ref{th:hahn} shows the fluctuations of $h_m$ are described by a CLT. One may interpret the result in the following way.  The centered height function $h_m-\EE h_m$ converges to a the Gaussian random distribution $F$. The distribution acts on $C^1$ test functions $f$ with compact support in $[0,1]$, under the pairing $\langle F, f'\rangle$ and the covariance structure is induced by the (Sobolev-type) norm $\sum_{k=1}^\infty k |\hat f_k|^2$.

It is interesting to compare this result with the results of \cite{Petrov} on the two dimensional fluctuations (see also \cite{Borad0,Borad,Borad2,Borad3} for two dimensional fluctuations in models related to classical orthogonal polynomial ensembles), where we view the height function as a function of two variables $(m,x)$.   It was proved in \cite{Petrov} that the two dimensional  fluctuations of the height function converge to the pull-back, by the complex structure,  of the Gaussian Free Field on the upper half plane with Dirichlet boundary conditions.   The complex structure is a diffeomorphism $\Omega$   that maps the disordered region in the hexagon onto the upper half plane.  
The  push-forward by $\Omega$ of  the centered fluctuation $h-\EE h$, viewed as a function of two variables, also converges to a  Gaussian random distribution $G$. The distribution now acts on sufficiently smooth test functions  $f$ with compact support in the upper half plane, under the pairing $-\langle G, \triangle f\rangle$ (where $\triangle$ is the Laplace operator) and the covariance structure is, up to a constant, induced by the Dirichlet norm $\int |\nabla f|^2 $.

Thus the centered height function converges to a random Gaussian distribution both as a function of two variables and as function of one variable along vertical sections. Naturally these results are strongly related since we probe the same random object with different classes of test functions, but we believe that it is not a priori obvious how to obtain one from the other. From this perspective, it is  interesting to compare the proof  for the one dimensional fluctuations in the present paper, with the proof for the two dimensional fluctuations of a  particular model in the paper \cite{D}. In both cases, the proof relies on the cumulant expressions for linear statistics, but the analysis is rather different.

\subsection{The two matrix model}\label{sec:two}

In the final example, we discuss the two matrix model. 
The two matrix model in random matrix theory is defined as the probability measure on the pairs $(M_1,M_2)$ of $n\times n$ Hermitian matrices defined by 
\beq \no
\frac{1}{Z_{2M}} {\rm e}^{-n \Tr \left(V_1(M_1)+V_2(M_2) -\tau M_1M_2\right)} \ {\rm d} M_1 {\rm d} M_2,
\eeq
where $\tau \neq 0$ is called the coupling constant, $Z_{2M}$ is a normalizing constant (partition function),  $V_1$ and $V_2$ are two polynomials of even degree and positive leading coefficients. Finally, ${\rm d}M_j$ stands for the product of the Lebesgue measure on the independent entries $${\rm d}M_j=\prod_{k=1}^n {\rm d}(M_j)_{kk} \prod_{1\leq k <l\leq n} {\rm d}\Re  (M_j)_{kl} {\rm d}\Im  (M_j)_{kl}.$$ 
The two matrix model gives rise to two biorthogonal ensembles in the following way: when we average over $M_2$ the eigenvalues of $M_1$ form a biorthogonal ensemble of size $n$ and vice versa. 

To describe the orthogonal ensembles we need the so-called biorthogonal polynomials that were introduced in \cite{EM}. Let   $\{p_{k,n}\}_k$ and $\{q_{l,n}\}_l$ be two families of polynomials such that $p_{k,n}$ and $ q_{l,n}$ have degree $k$ and $l$ respectively, and 
\beq \label{eq:biorthogonal2M}
\iint_{\bbR^2} p_{k,n}(x) q_{l,n}(y) {\rm e}^{-n \left(V(x)+W(y)-\tau x y \right)} {\rm d} x {\rm d} y= \delta_{kl}.
\eeq
It should be noted that  this orthogonality is not Hermitian and we can therefore not apply Gram-Schmidt to ensure the existence of such polynomials. Nevertheless it can be proved that they exist and have simple real zeros  \cite{ErclMcL}  and that the zeros have an interlacing property \cite{DGK}.  Finally, the polynomials $p_{k,n}$ and $q_{l,n}$ are unique up to a multiplicative constant (see also Remark \ref{rem:changeinJ}). We fix this constant by assuming that the leading coefficients of $p_{k,n}$ and $q_{k,n}$ are the same. We note that this convention is different from, for example, \cite{DGK,DKM} but fits our purposes better.

We will also need the (Laplace) transforms 
\beq \label{eq:laplacetransforms}
\begin{split}
P_{k,n} (y) & = {\rm e}^{-n W(y) } \int p_{k,n}(x) {\rm e}^{-n (V(x)-\tau x y)} \ {\rm d}x, \\
Q_{l,n}(x) &= {\rm e}^{-n V(x) } \int q_{k,n}(y) {\rm e}^{-n (w(y)-\tau x y)} \ {\rm d}y.
\end{split}
\eeq
With these functions we can rewrite \eqref{eq:biorthogonal2M} as 
\beq\label{eq:bioortho2}
\begin{split}
\int_\bbR p_{k,n}(x) Q_{l,n} (x) {\rm d} x&= \delta_{lk},\\
\int_\bbR P_{k,n}(y) q_{l,n}(y) {\rm d} y&= \delta_{lk}.
\end{split}
\eeq
Thus the eigenvalues of $M_1$ when averaged over $M_2$, form a biorthogonal ensemble generated by $\{p_{k,n}\}_k$ and $\{Q_{k,n}\}_k$ with kernel 
\beq\no
K^{(1)}_n(x_1,x_2)= \sum_{k=0}^{n-1} p_{k,n}(x_1) Q_{k,n}(x_2).
\eeq 
And vice versa,  the eigenvalues of $M_2$ when averaged over $M_1$, form a biorthogonal ensemble generated by $\{q_{l,n}\}_l$ and $\{P_{l,n}\}_l$ with kernel 
\beq\no
K^{(2)}_n(y_1,y_2)= \sum_{l=0}^{n-1} P_{l,n}(y_1) q_{l,n}(y_2).
\eeq 
Given a function $f$ we denote the associated linear statistic for the biorthogonal ensemble with kernel $K_1$ by $X^{(n)}_f$. The associated linear statistic for the biorthogonal ensemble with kernel $K_2$ will be denoted by $Y^{(n)}_f$. 

Now that we have established two biorthogonal ensembles, the question rises if they admit a recurrence. This answer can be found in \cite{BE}, where the authors provide a thorough study of the integrable structure of the biorthogonal polynomials. The fact that we have a recurrence is not hard to prove. Indeed, it is trivial  that $x p_{k,n}(x)$ can be expressed as a linear combination of $p_{m,n}(x)$ for $m=0,\ldots,k+1$. By using the orthogonality we see that  the precise coefficients can be obtained by computing the following integral
\begin{multline}\label{eq:integrationbyparts}
\iint xp_{k,n}(x)q_{m,n}(y) \exp(-n (V_1(x)+V_2(y)-\tau xy )) {\rm d} x{\rm d} y\\
=\frac{1}{\tau} \iint p_{k,n}(x)\left(V_2'(y) q_{m,n}(y)-q'_{m,n}(y)\right)  \exp(-n (V_1(x)+V_2(y)-\tau xy )) {\rm d} x{\rm d} y,
\end{multline}  
where we used integration by parts in the last step.
Since $V_2$ is a polynomial, the latter integral vanishes if the degree of $V_2'(y) q_{m,n}(y)$ (and hence also $m$) is less than $k$. Concluding, there exist coefficients $a^{(k)}_{m,n}$ and $b^{(l)}_{m,n}$ such that 
\beq\label{eq:biorecur1}
\begin{split}
x p_{k,n}(x)&=\sum_{m=0}^{d_2} a^{(k)}_{m-1,n}  p_{k-m+1,n} (x) \\
y q_{l,n}(y)&=\sum_{m=0}^{d_1} b^{(l)}_{m-1,n} q_{l-m+1,n} (y) 
\end{split}
\eeq
where $d_j= \mathop{\mathrm{deg}} V_j$.  Note that our convention that $p_{k,n}$ and $q_{k,n}$ have the same leading coefficients implies that  $a^{(k)}_{-1,n}=b^{(k)}_{-1,n}$. 
\begin{proposition}
Assume that there exists $a_{m}$ and a subsequence $\{n_j\}$ such that 
\beq \label{eq:condition2Ma}
\lim_{j\to \infty}a_{m,n_k}^{n_k+r} = a_m, \qquad m=-1,\ldots,d_2-1, \quad r\in \bbZ.
\eeq
Define $s(w)=\sum_{l=-1}^{d_2-1} a_l w^{l}$.  Then for any polynomial $p$ with real coefficients  we have that 
\beq\no
X_p^{(n_j)}-\EE X_p^{(n_j)}\to N\left(0,\sum_{k=1}^\infty  k \hat p_k  \hat p_{-k} \right),
\eeq
where 
\beq\no
\hat p_k = \frac{1}{2\pi {\rm i}} \oint_{|w|=1} p(s(w)) \frac{{\rm d} w}{w^{k+1}}. 
\eeq
Similarly, assume that there exists $b_{m}$ and a subsequence $\{n_j\}$ such that 
\beq\label{eq:condition2Mb}
\lim_{j\to \infty}b_{m,n_j}^{n_k+r} = b_m, \qquad m=-1,\ldots,d_1-1, \quad r\in \bbZ.
\eeq
Define $s(w)=\sum_{l=-1}^{d_1-1} b_l w^{l}$.  Then for any polynomial $p$ with real coefficients we have that 
\beq\no
Y_p^{(n_j)}-\EE Y_p^{(n_j)}\to N\left(0,\sum_{k=1}^\infty k \hat p_k  \hat p_{-k} \right),
\eeq
where 
\beq\no
\hat p_k = \frac{1}{2\pi {\rm i}} \oint_{|w|=1} p(s(w)) \frac{{\rm d} w}{w^{k+1}}. 
\eeq
\end{proposition}
The natural question arises as to what conditions on $V_1$, $V_2$ and $\tau$ imply \eqref{eq:condition2Ma} and/or \eqref{eq:condition2Mb} . It is expected that these limits hold under certain conditions on the potential $V_1$ and $V_2$ (in particular, when the spectral curve has genus $0$). However, a rigorous asymptotic analysis for the two matrix model has only been carried out  \cite{DKM} for  the biorthogonal ensemble with kernel $K^{(1)}$ when $V_2$ is quartic. 

We also want to remark that the recurrence coefficients in \eqref{eq:biorecur1} are related and that the \eqref{eq:condition2Ma} and \eqref{eq:condition2Mb} always come hand in hand. Indeed, if we define the matrices $ J_1$ and $ J_2$ by 
\beq\no
\begin{split}
\left(J_1\right)_{kl}&=\iint_{\bbR^2} x p_{k,n}(x)q_{l,n}(y) {\rm e}^{-n \left(V_1(x)+V_2(y)-\tau x y )\right)} {\rm d} x{\rm d} y,\\
\left( J_2\right)_{kl}&=\iint_{\bbR^2} y p_{l,n}(x)q_{k,n}(y) {\rm e}^{-n \left(V_1(x)+V_2(y)-\tau x y )\right)} {\rm d} x{\rm d} y,\\
\end{split}\eeq
then it
 follows from \eqref{eq:integrationbyparts} that 
\beq\no
\begin{split}
\left(J_1\right)_{kl}= \left(V_2'(J_{{2}})\right)_{lk},\quad k \geq l,\\
\left(J_2\right)_{kl}= \left(V_1'(J_{{1}})\right)_{lk},\quad k \geq l,
\end{split}
\eeq
Combining this with $a^{(k)}_{-1,n}= b^{(k)}_{-1,n}$ we see that the values of $a^{(k)}_{m,n}$ are fully determined by the values of $ b^{(k)}_{m,n}$ and vice versa. Hence if $J_1$ has a right limit, then so has $J_2$ and vice versa.  

The relation between $J_1$ and $J_2$ can be effectively used in the case  $V_1(x)=x^2$.   In this case, it is not difficult to simplify the transform in \eqref{eq:laplacetransforms} using gaussian integrals and obtain 
\beq\no
P_{k,n}(y)= {\rm e}^{-n(V_2(y)-\tau^2 y^2/2) } \tilde q_{k.n}(y),
\eeq
for some polynomial $\tilde q_{k,n}$ of degree $k$. But the the orthogonality in \eqref{eq:bioortho2} tells us that $\tilde q_{k,n}=q_{k,n}$ is the orthonormal polynomial of degree $k$ with respect to the varying measure $ {\rm e}^{-n(V_2(y)-\tau^2 y^2/2) } $ on  $\bbR$. Hence the biorthogonal ensemble with kernel $K^{(2)}$ is in fact an orthogonal polynomial ensemble.  Since the asymptotic behavior of these recurrence coefficients is fairly well understood, we can use these results to get information on $J_1$. This idea has been used in \cite{DGK} where a vector equilibrium problem  for the limiting density of points was derived.

\begin{theorem}
Assume that $V_1(x)=x^2/2$ and assume that 
\beq \no
\lim_{j\to \infty} a^{(n_j+r)}_{-1,n}=\lim_{j \to \infty} a^{(n_j+r)}_{1,n}=a, \qquad \lim_{j\to \infty} a^{(n_j+r)}_{0,n}=b, \qquad r\in \bbZ.
\eeq
for some subsequence $\{n_j\}$. 
Define $s_1(w)=\frac1\tau\left(V'_2(aw +b +a/w)\right)_-+a w$ and  $s_2(w)=a w+b +a/w$. 
Then for any polynomial $p$ with real coefficients  we have that 
\beq\no
\begin{split}
X_p^{(n_j)}-\EE X_p^{(n_j)}&\to N\left(0,\sum_{k=1}^\infty  k \hat p_k  \hat p_{-k} \right),\\
Y_p^{(n_j)}-\EE Y_p^{(n_j)}&\to N\left(0,\sum_{k=1}^\infty  k \hat {\hat p}_k  \hat {\hat p}_{-k} \right),
\end{split}
\eeq
where 
\beq\no
\begin{split}
\hat p_k &= \frac{1}{2\pi {\rm i}} \oint_{|w|=1} p(s_1(w)) \frac{{\rm d} w}{w^{k+1}},\\
\hat {\hat p}_k &= \frac{1}{2\pi {\rm i}} \oint_{|w|=1} p(s_2(w)) \frac{{\rm d} w}{w^{k+1}},\\
\end{split}
\eeq
\end{theorem}

\end{document}